\documentclass[10pt]{iopart}
\usepackage{latexsym,epsfig,graphicx,bbm,bm,amssymb,wasysym}
\usepackage{amsthm}

\theoremstyle{definition}
\newtheorem{definition}{Definition}

\newtheorem{theorem}{Theorem}

\newtheorem{lemma}{Lemma}

\theoremstyle{remark}

\usepackage[export]{adjustbox}
\usepackage[framemethod=tikz]{mdframed}

\usepackage{hyperref}

\newcommand{\ket}[1]{\left | \, #1 \right \rangle}
\newcommand{\kets}[1]{| \, #1 \rangle}

\newcommand{\eqr}[1]{Eq.~(\ref{#1})}
\newcommand{\fir}[1]{Fig.~\ref{#1}}
\newcommand{\secr}[1]{Sec.~\ref{#1}}

\newcommand{\half}{\mbox{$\textstyle \frac{1}{2}$}}
\newcommand{\bra}[1]{\left \langle #1 \, \right |}
\newcommand{\braket}[2]{\left\langle\, #1\,|\,#2\,\right\rangle}

\newcommand{\av}[1]{\langle #1\rangle}

\begin{document}

\title[Compact Neural-network Quantum States]{Compact Neural-network Quantum State representations of Jastrow and Stabilizer states}

\author{Michael Y. Pei$^{1}$ and Stephen R. Clark$^{1}$}
\address{$^1$H.H. Wills Physics Laboratory, University of Bristol, Bristol BS8 1TL, UK.}
\ead{stephen.clark@bristol.ac.uk}

\begin{abstract}
Neural-network quantum states (NQS) have become a powerful tool in many-body physics. Of the numerous possible architectures in which neural-networks can encode amplitudes of quantum states the simplicity of the complex restricted Boltzmann machine (RBM) has proven especially useful for both numerical and analytical studies. In particular devising exact NQS representations for important classes of states, like Jastrow and stabilizer states, has provided useful clues into the strengths and limitations of the RBM based NQS. However, current constructions for a system of $N$ spins generate NQS with $M \sim O(N^2)$ hidden units that are very sparsely connected. This makes them rather atypical NQS compared to those commonly generated by numerical optimisation. Here we focus on {\em compact} NQS, denoting NQS with a hidden unit density $\alpha = M/N \leq 1$ but with system-extensive hidden-visible unit connectivity. By unifying Jastrow and stabilizer states we introduce a new exact representation that requires at most $M=N-1$ hidden units, illustrating how highly expressive $\alpha \leq 1$ can be. Owing to their structural similarity to numerical NQS solutions our result provides useful insights and could pave the way for more families of quantum states to be represented exactly by compact NQS.\\

\end{abstract}


\maketitle

\section{Introduction} 

The quantum many-body problem encapsulates numerous fascinating and technologically relevant phenomena even in its simplest instance of localised spins, where effective interacting lattice Hamiltonians can give rise to antiferromagnetism, frustration and topological spin liquids~\cite{zhou17,savary17}. Capturing these collective many-body effects in numerical calculations is a striking challenge due to the exponential growth of their Hilbert space $\mathbbm{C}^2\otimes\mathbbm{C}^2\otimes \cdots \otimes\mathbbm{C}^2$ with system size. One strategy to overcome this `curse of dimensionality' is the {\em variational method} where an approximate representation of a system's ground state is found within a class of trial states defined by a tractable number of parameters~\cite{foulkes01,gubernatis16,becca17}. To be successful some physical insight is often needed when devising trial states, so they can capture expected correlations, while also having a functional form that allows key observables, like energy, to be efficiently computed. The major issue is that variational calculations are inevitably biased by the trial state and can inadvertently predict the wrong physics. To avoid this, sophisticated trial states are needed whose expressibility, and thus number of parameters, can be systematically enlarged and in principle can even become exact in some extreme limit. 

A powerful example of such a variational approach are tensor network states~\cite{schollwock11,verstraete08,cirac09,orus14}. These work by decomposing the amplitudes of a quantum state into a network of connected multidimensional arrays of complex numbers that form the variational parameters of the trial state. Tensor network states have several highly useful properties. First, efficient deterministic numerical algorithms have been devised for their direct optimisation. Second, the number of variational parameters, governed by the size of the tensors and the geometry of their connections, can be directly related to the entanglement between subsystems of the state, which often obeys an ``area law"~\cite{eisert10}. Consequently tensor networks possess a rather broad form of variational bias to low entanglement states. These properties have proven to be extremely effective for systems in one spatial dimensional (1D) where matrix product states (MPS)~\cite{schollwock11,orus14} have allowed a vast range of physics to be accessed with unprecedented accuracy~\cite{schollwock05}. For 2D systems projected entangled pair states (PEPS)~\cite{verstraete08,cirac09}, tree tensor networks (TTN)~\cite{shi2006,murg10} and multiscale entanglement renormalization ansatz (MERA)~\cite{evenbly13} have attempted to replicate this success. However, the scaling of the computational cost of these tensor network algorithms with the number of variational parameters, while polynomial and hence formally efficient, are nonetheless severe enough to render accurate calculations extremely demanding. It is an active field of research to devise schemes to reduce this cost~\cite{fishman18,zaletel20}.

In recent years there has been considerable interest in adapting techniques from machine learning to help tackle many-body physics problems~\cite{carleo19,torlai20,carrasquilla20}. A popular approach is to use generative models based on artificial neural networks (ANN) which have been especially successful in taming the curse of dimensionality for conventional `big data' problems, allowing complex patterns and abstractions to be identified~\cite{goodfellow16}. This ability has strongly motivated applying ANN to efficiently encode quantum many-body states. A direct way to accomplish this, introduced by Carleo and Troyer~\cite{carleo_nqs17}, is to view an ANN itself as many-body ansatz wavefunction in which the couplings between neurons act as complex variational parameters that are stochastically optimised. This neural-network quantum state (NQS) approach is applicable to any number of spatial dimensions and can leverage the differing strengths of the numerous variants of ANN known in machine learning. After the original work~\cite{carleo_nqs17} based on a restricted Boltzmann machines (RBM) architecture~\cite{fischer12} with complex parameters, subsequent studies investigated their generalisation to deep Boltzmann machines (DBM)~\cite{gao_dbm17,carleo_dbm18,he_mdbm19}, as well as feedforward~\cite{saito18,choo18,luo19,adams2020}, convolution~\cite{choo_fmag19,naoki2020,markus2020,liang2021} and recurrent neural networks~\cite{levine19}.

In this work we focus exclusively on NQS based on the venerable RBM architecture for several reasons. First, they have a simple shallow structure comprising only two layers of neurons, a visible layer with $N$ units and a hidden layer with $M$ units. Second, it is well known they can capture arbitrary functions once $M$ scales exponentially with $N$~\cite{leroux08}, meaning RBMs are an exhaustive ansatz and suggests they could be an effective weakly biased ansatz for approximations with a tractable $M$. Third, RBMs can host states exhibiting volume-law entanglement scaling~\cite{deng17} indicating a quite distinct representational power compared to tensor networks, despite the rich conceptual connections between them~\cite{clark_cps18,chen18,collura20}. Finally, since they have no intralayer couplings RBMs allow for efficient sampling making them extremely well suited to stochastic optimisation within variational Monte Carlo (VMC)~\cite{foulkes01,gubernatis16,becca17}. 

As such NQS wavefunctions have been applied to wide-ranging problems including frustrated spin systems~\cite{westerhout20,nomura2021}, interacting fermi~\cite{choo20} and bose systems~\cite{saito17,mcbrian19}, simulating quantum circuits~\cite{jonsson18,freitas18,bausch20}, as well as describing states with topological order~\cite{kaubruegger18,glasser18,huang17} and non-Abelian symmetries~\cite{vieijra2020}. Moreover, they have proved to be very effective at enhancing~\cite{clark_cps18} more traditional pair product wave functions in a hybrid approach for both fermionic~\cite{nomura17} and spin systems~\cite{ferrari19}, and have also been successfully generalised for open quantum systems~\cite{torlai18,vicentini19,hartmann19,nobuyuki19}. Complementary to numerical studies, which essentially treat NQS as a {\em blackbox} of amplitudes, are analytical studies directly examining their expressive power by constructing exact representation of broad and relevant classes of quantum states. 

\begin{figure}[htb]
\begin{center}
\includegraphics[scale=0.5]{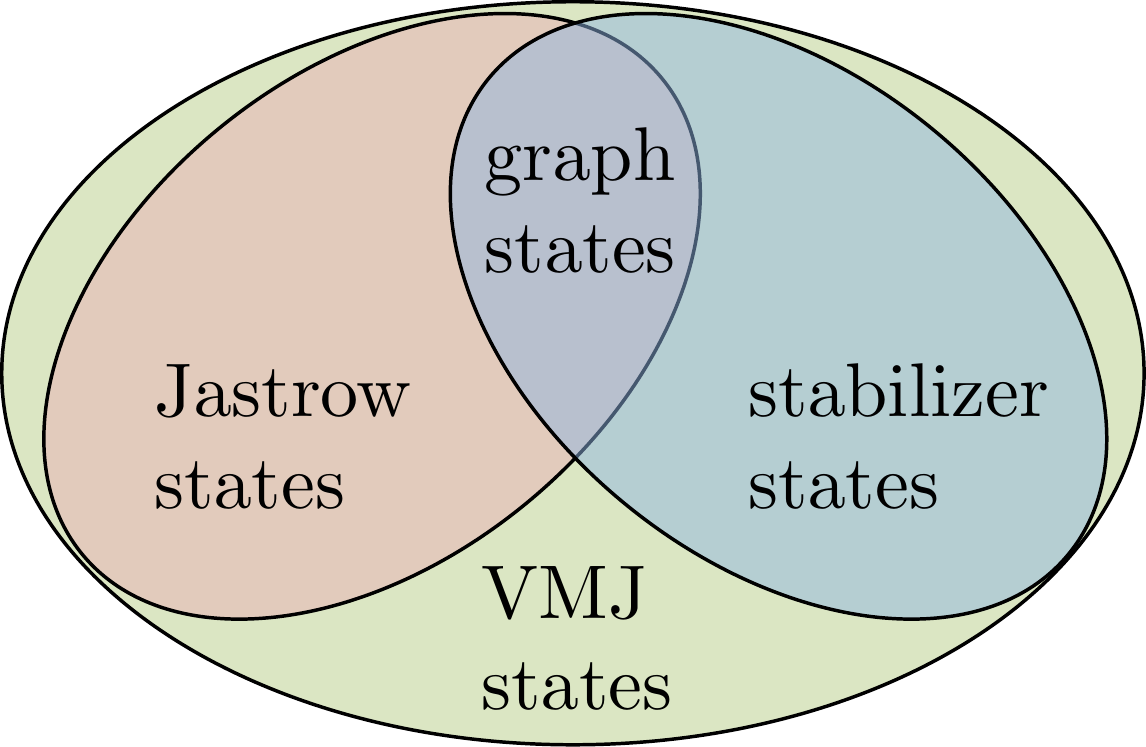}
\end{center}
\caption{A schematic of the relation between Jastrow and stabilizer states. Graph states lie within the intersection of these two distinct classes of states. By applying graph theoretic tools to Jastrow states we propose a new larger class of states called {\em vertex modified Jastrow} (VMJ) states in which arbitrary single-spin gates are applied specific spins. Although a modest generalisation of Jastrow we show this is sufficient to capture stabilizer states and provide a simple procedure for constructing compact NQS.}
\label{fig:states}
\end{figure}

Initial work in this direction focused on important specific cases like the Toric code states and a symmetry-protected cluster state~\cite{deng17_top} to show how topological order can be captured, as well as specially designed graph/stabilizer states to demonstrate analytically volume-scaling entanglement~\cite{deng17}. A construction to represent any graph state as an NQS then followed~\cite{gao_dbm17}. This was then generalised to weighted graph states~\cite{clark_cps18}, spin Jastrow states, which encompass all (weighted) graph states including Laughlin-like states describing chiral topological phases~\cite{glasser18,kaubruegger18,clark_cps18}, as well as the entire class of stabilizer states~\cite{zheng_stbl19,zhang_rbm2stbl18,lu19,jia_surf19} containing quantum error correcting codes like the Toric code states. Even more exotic forms of topological ordering based on hypergraph states and XS-stabilizer states have also been shown to have exact NQS representations~\cite{lu19}. All these studies have enriched our understanding of how features like the sign and nodal structure as well as non-local correlations can be encoded within NQS.

As illustrated in \fir{fig:states}, Jastrow states~\cite{jastrow55} and stabilizer states~\cite{gottesman97} both form large distinct classes of states defined in general by $O(N^2)$ parameters and whose intersection contains graph states~\cite{hein06}. Currently known explicit NQS constructions encode these states in $M \sim O(N^2)$ hidden units with the sparsest possible connections~\cite{gao_dbm17,glasser18,kaubruegger18,zheng_stbl19,zhang_rbm2stbl18,lu19,jia_surf19}. Yet the expressiveness of NQS depends sensitively on the pattern of connections, so numerical optimisations often exploit the most general hidden units possessing system-extensive connectivity. The implications of these exact representations for numerical calculation is therefore unclear. One might attempt to improve the efficiency by directly mimicking their specialised sparcity structure, although this nullifies the flexibility of NQS that originally motivated them and introduces significant bias. Conversely, if these exact constructions do genuinely represent the sparcity and complexity of hidden units needed for Jastrow and stabilizer states, then they also suggest that they are rather non-typical states compared to those that commonly emerge from numerical calculations. Moreover, it would mean that if this sparcity structure was not known in advance, then these classes of states are prohibitively expensive to numerically ``learn" since the optimisation problem would begin with an unfavourable $O(N^3)$ number of parameters~\cite{glasser18}. However, the non-uniqueness of NQS representations together with a parameter counting argument strongly suggest that Jastrow and stabilizer states should be representable exactly by an NQS with only $M \sim O(N)$ system-extensive hidden units.

In this paper we confirm this conjecture. Our result is built using vertex modified Jastrow (VMJ) states, which are a new modest generalisation of Jastrow states in which a specific subset of spins can have arbitrary single-spin gates applied to them. Crucially, the enlarged class of VMJ states encompasses stabilizer states, as shown in \fir{fig:states}. By combining graph theoretic concepts with tensor network diagrammatics for NQS states~\cite{clark_cps18} we show how to construct NQS representations of VMJ states with at most $M = N-1$ hidden units. The system-extensive connectivity of this new NQS representation enhances its relevance for numerical calculations and highlights the role of strong visible-hidden unit correlations. In particular it indicates that NQS optimisations with a very low hidden unit density $\alpha = M/N < 1$ should be able to capture and improve on Jastrow states, thereby providing an efficient and systematic way to go beyond this commonly used variational class. Indeed this dramatic ability of $\alpha < 1$ NQS to exactly represent complex classes of quantum states resonates with recent work~\cite{abanin2020} revealing their ability to capture several orders of a perturbative expansion. We coin the term {\em compact} NQS for $\alpha \leq 1$ to highlight this acute compression in hidden unit complexity. 

The structure of this paper is as follows. In \secr{sec:nqs} we give an overview of the many-body problem, the variational method and NQS in their original complex RBM formulation. In \secr{sec:jastrow} we discuss the conventional $M\sim O(N^2)$ construction of NQS for Jastrow states, before modifying it to arrive at a compact NQS based on perfect hidden-visible correlations and demonstrating its emergence via numerical optimisation for the one-dimensional XXZ spin chain. In \secr{sec:tnt} we outline diagrammatic tools that allow NQS to be recast as a tensor network and identify a number of properties, most crucially when arbitrary single-spin gates can be readily absorbed into an NQS tensor network. In \secr{sec:vmj} we use the tensor network approach to identify a canonical form for Jastrow state NQS, introduce some graph theoretic concepts, and use these to define VMJ states and their general compact NQS representation. In \secr{sec:graph} we review graph and stabilizer states, prove that any stabilizer state can be described as a VMJ-NQS, and illustrate this analytically for several important special cases. Finally, in \secr{sec:conclusion} we conclude and discuss some open problems.

\section{Neural-network quantum states} \label{sec:nqs}
In this section we briefly introduce the quantum many-body problem and VMC, before reviewing NQS in terms of complex RBM as a powerful ansatz for this approach.  

\subsection{Quantum many-body problem and variational approach}
In this work we will focus on physical systems composed of $N$ spin-$\half$ particles each described by a vector of Pauli operators $(\hat{X}_j,\hat{Y}_j,\hat{Z}_j)$ for $j=1,2,\dots,N$ and a local basis $\hat{Z}_j\ket{v_j} = v_j\ket{v_j}$, where $v_j \in \{+1,-1\}$. The $z$ basis for the full system is then $\ket{\bm v} =\ket{v_1}\otimes \cdots \otimes \ket{v_N}$ where ${\bm v} = (v_1,v_2,\dots,v_N) \in \{+1,-1\}^N$, and any many-body quantum state can be expanded in this basis as
\begin{equation}
\ket{\Psi} = \sum_{\bm v} \Psi({\bm v}) \ket{\bm v},
\end{equation}
via its complex amplitudes $\Psi({\bm v})$. It will also be convenient on some occasions to label the $z$ basis instead using a `qubit' binary string ${\bm q} = (q_1,q_2,\dots,q_N) \in \{0,1\}^N$ where ${\bm v} = (-1)^{\bm q} = 1 - 2{\bm q}$.

Our central problem is to solve the eigenvalue problem $\hat{H}\ket{\Psi_0} = E_0\ket{\Psi_0}$ for the ground state $\ket{\Psi_0}$ of a system governed by an interacting Hamiltonian $\hat{H}$ comprised of products of the Pauli operators over two or more spins. Given that the expectation value of any observable $\hat{A}$ for a general (unnormalised) state $\ket{\Psi}$ is 
\begin{equation}
\av{A}_\Psi = \frac{\bra{\Psi}\hat{A}\ket{\Psi}}{\braket{\Psi}{\Psi}},
\end{equation}
the variational approach reformulates this problem as the minimisation of the energy $E = \av{\hat{H}}_\Psi$ over the exponentially many amplitudes $\Psi({\bm v})$. Performing this task exactly is only feasible for small systems $N \sim O(10)$~\cite{lauchli12}. 

The curse of dimensionality is circumvented by instead restricting the optimisation over a specialised class of states $\ket{\Psi_{\bm p}}$, dependent on parameters $\bm p$, whose number scales polynomially with $N$. The variational principle $E_0 \leq E_{{\bm p}_0} = {\rm min}_{\bm p} \av{\hat{H}}_{\Psi_{\bm p}}$ is then used to determine the ``best'' approximation $\ket{\Psi_{{\bm p}_0}}$ within the ansatz for $\ket{\Psi_0}$. Typically even for simple observables $\av{A}_{\Psi_{\bm p}}$ cannot be evaluated exactly from variational ansatzes for many-body systems~\cite{gubernatis16,becca17} and so Monte Carlo sampling is employed giving the VMC approach~\cite{foulkes01}. A crucial feature of VMC is that only ratios of amplitudes for an ansatz $\Psi_{\bm p}({\bm v})/\Psi_{\bm p}({\bm v}')$ between different spin states $\ket{\bm v}$ and $\ket{{\bm v}'}$ are needed for Markov chain sampling, so we can safely ignore the normalisation of quantum states in this work. 

\subsection{Restricted Boltzmann machine formulation} \label{sec:rbm}
Neural-network quantum states are a promising approach for constructing highly efficient and accurate approximate representations of the exponentially many amplitudes $\Psi({\bm v})$. The structure of NQS is motivated from classical probabilistic models called RBMs commonly used in machine learning~\cite{carleo_nqs17}. In these models there are two sets of units, $N$ {\em visible} units representing the system itself, and $M$ {\em hidden} units which are additional degrees of freedom to be marginalised. Amplitudes of NQS then follow from a complex Boltzmann-like ansatz
\begin{equation}
\Psi_{\bm \lambda}({\bm v}) = \sum_{{\bm h}} \exp\left[\mathcal{E}_{\bm \lambda}({\bm v},{\bm h})\right], \label{rbm}
\end{equation}
characterised by a classical ``energy" function
\begin{equation}
\mathcal{E}_{\bm \lambda}({\bm v},{\bm h}) = \sum_{j=1}^N a_j v_j + \sum_{i=1}^M b_i h_i + \sum_{i=1}^M\sum_{j=1}^N w_{ij}h_i v_j, \label{rbm_energy}
\end{equation}
describing the interactions between the visible units, with the physical configuration ${\bm v}$, and hidden units, with configuration ${\bm h} = (h_1,\dots,h_M)$. The geometry of an RBM is shown in \fir{fig:rbm}. An NQS is thus defined by the $MN + M + N$ complex parameters ${\bm \lambda} = \{{\bm a},{\bm b},{\bm w}\}$ in the energy function, comprising an $N$-dimensional vector ${\bm a}$ of visible biases, an $M$-dimensional vector ${\bm b}$ of hidden biases, and $M \times N$ matrix ${\bm w}$ of weights. Assuming the hidden units are also described by Ising-like variables $h_j \in \{+1,-1\}$ then after tracing over them we obtain the well-known amplitude expansion~\cite{carleo_nqs17}
\begin{eqnarray}
\Psi_{\bm \lambda}({\bm v}) &=& \prod_{j=1}^N e^{a_jv_j}  \prod_{i=1}^M 2\cosh\left(b_i + \sum_{j=1}^N w_{ij}v_j\right), \label{rbm_amps}
\end{eqnarray}
which is particularly well-suited to variational optimisation.

\begin{figure}[ht]
\begin{center}
\includegraphics[scale=0.5]{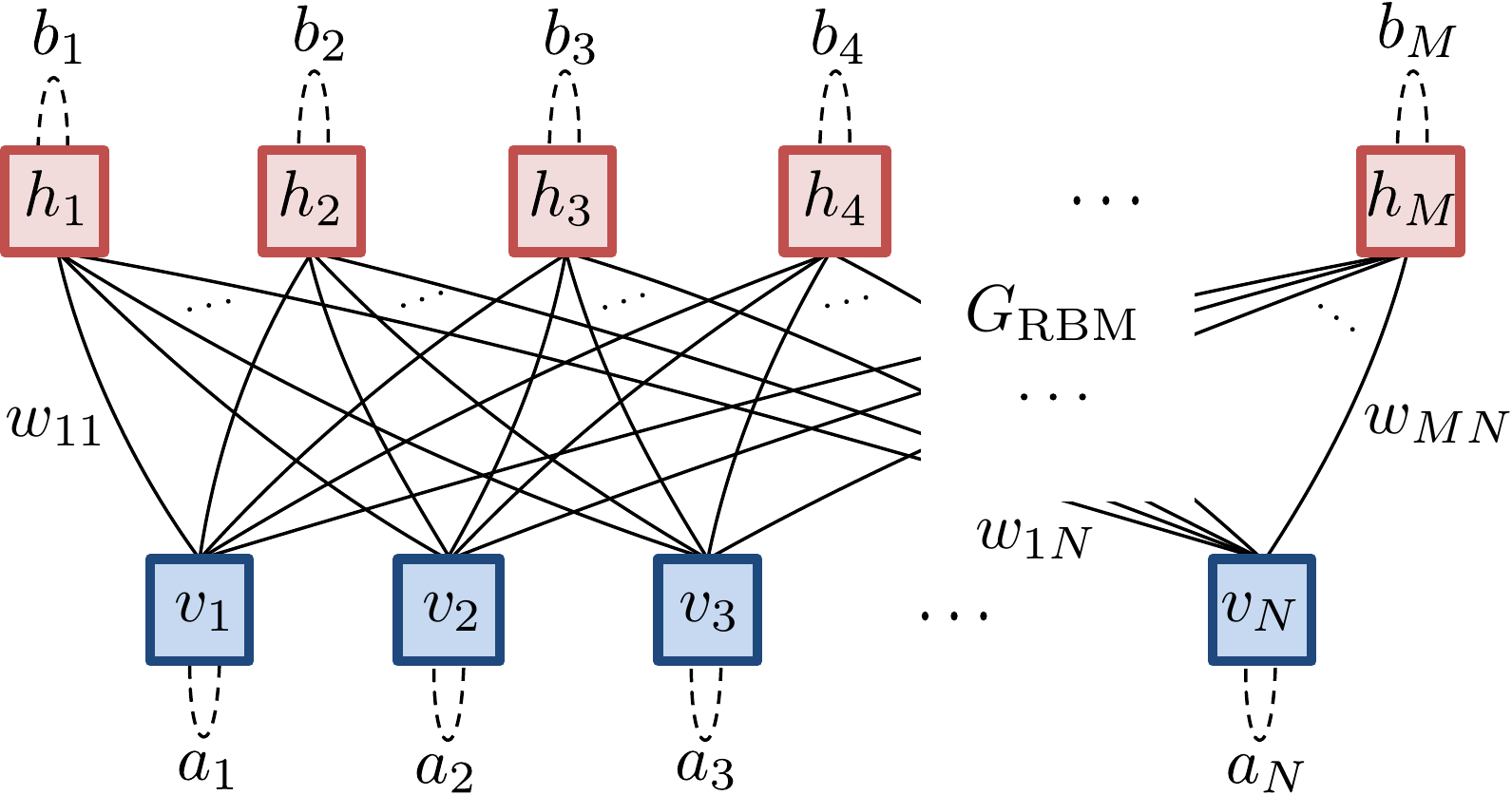}
\end{center}
\caption{The bipartite graph $G_{\rm RBM}$ of interaction weights $\bm w$ between hidden and visible units in an RBM with edges shown as solid arcs (see \secr{sec:jastrow}). For completeness the biases $\bm a$ and $\bm b$ on each unit are depicted here as self-loop edges, but are given by dotted arcs since these are not included in a simple RBM graph.}
\label{fig:rbm}
\end{figure}

It will prove useful for later that we introduce the following terminology:
\theoremstyle{definition}
\begin{definition}[Univalent visible unit]
The number of hidden units connected to a given visible unit in an NQS is called its {\em valency}. Any visible unit connected to a single hidden unit is said to be {\em univalent}. The number of such visible units in an NQS is the {\em univalency} of the representation.
\end{definition}
Although \eqr{rbm} is based on a complex Boltzmann-like form with $\mathcal{E}_{\bm \lambda}({\bm v},{\bm h})$ restricted to pairwise visible-hidden interactions only, tracing out the hidden units generates an effective energy function for the visible units alone
\begin{equation}
\mathcal{E}_{{\rm eff},\bm \lambda}({\bm v}) = \log\left[\sum_{\bm h}e^{\mathcal{E}_{\bm \lambda}({\bm v},{\bm h})}\right], \label{eq:eff_energy}  
\end{equation}
that can contain much more complicated higher-order interactions. This expressiveness is controlled by increasing the number $M$ of hidden units allowing for complex correlations within $\Psi_{\bm \lambda}({\bm v})$ to be encoded. The NQS ansatz is exhaustive in the sense that for $M \sim 2^N$ arbitrary states $\Psi({\bm v})$ can be described exactly~\cite{leroux08}. However, practical NQS representations instead have scaling $M \sim {\rm poly}(N)$ which can be sampled in a formally efficient way, including exact representations of Jastrow~\cite{glasser18,kaubruegger18,clark_cps18}, graph~\cite{gao_dbm17}, stabilizer~\cite{zheng_stbl19,zhang_rbm2stbl18,lu19,jia_surf19}, hypergraph and XS-stabilizer states~\cite{lu19}. Numerical calculations have demonstrated that many important interacting systems have ground states that can be well approximated by an efficient NQS with $\alpha \sim 32$ feasible with VMC~\cite{nomura20}. However, if such calculations involve states where $M$ scales superlinear in $N$ then they can quickly become impractical. This strongly motivates understanding NQS that have exact representations with $M\sim O(N)$, and even more strictly compact NQS defined as: 

\theoremstyle{definition}
\begin{definition}[Compact NQS]
An exact NQS representation of a state for a system size of $N$ is said to be {\em compact} if it requires $M \leq N$ hidden units.
\end{definition}

\section{Jastrow states} \label{sec:jastrow}
One of the most well-established and intuitive variational ansatz are so-called Jastrow states. Originally devised for wavefunctions in the continuum~\cite{jastrow55}, Jastrow states work by imprinting pairwise correlations on top of a reference state of independent particles. Indeed, one of the earliest VMC calculations performed for a many-body quantum system utilised Jastrow states to study Lennard-Jones interacting bosons modelling liquid $^4$He~\cite{mcmilan65}. For discrete systems like spins a general definition of Jastrow states involves a {\em graph}. Specifically, a graph $G$ is defined by the set $\mathcal{V}$, which is a finite-sized subset of $\mathbbm{N}$ labelling vertices, and a set $E(G)$ comprising 2-element subsets $(a,b)$ of $\mathcal{V}$ labelling edges~\cite{bondy2011}. We will only consider undirected edges so the order of vertices in $(a,b)$ is of no relevance. However, when summing over elements of $E(G)$ it is convenient to use the convention that $({\tt s},{\tt t})$ denotes the source $\tt s$ and target $\tt t$ vertices that obey ${\tt s} < {\tt t}$. There is a one-to-one correspondence between a graph $G$ and its symmetric binary $|\mathcal{V}| \times |\mathcal{V}|$ adjacency matrix ${\bm \Theta}$ defined as $\Theta_{ab} = 1$ if $\{a,b\} \in E$ and $\Theta_{ab} = 0$ otherwise. We will draw a graph with vertices as bordered squares and edges by arcs joining pairs of them, for example as
\begin{equation}
\includegraphics[scale=0.5,valign=c]{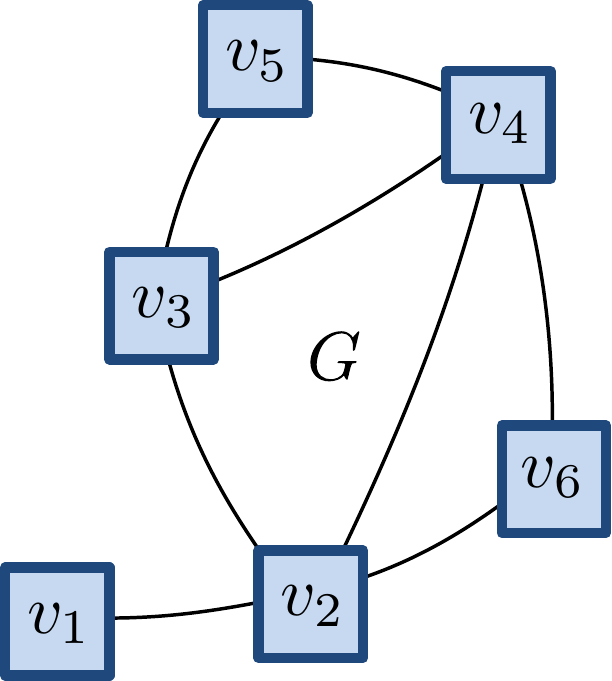} , \label{eq:simple_graph}
\end{equation}
and as seen earlier in \fir{fig:rbm} where we had two classes of vertices, visible and hidden. In the following we will consider simple graphs, so they contain no self-loops or multiple edges between the same vertices. In addition to this, and without loss of generality\footnote{The methods we will use can be applied individually to any subgraphs of a fragmented graph.} we will concentrate on connected graphs where any two vertices $a$ and $b$ are always connected by at least one sequence of edges (e.g. a path) in $G$. 

Given a graph $G$ over $|\mathcal{V}| = N$ visible units Jastrow states have amplitudes that are conveniently parameterised in a complex Boltzmann-like form as
\begin{equation}
\psi_{\rm JS}({\bm v}) = \exp\left(\mathcal{E}_{{\rm JS},{\bm \eta}}\right), \label{eq:jastrow_amps}
\end{equation}
in terms of an energy function 
\begin{equation}
\mathcal{E}_{{\rm JS},{\bm \eta}} = \sum_{j=1}^N c_jv_j + \sum_{i=1}^{|E(G)|} V_{{\tt s}_i {\tt t}_i}v_{{\tt s}_i} v_{{\tt t}_i}. \label{eq:jastrow_energy}
\end{equation}
Consequently, a Jastrow state is defined by the $N+|E(G)|$ complex parameters ${\bm \eta} = \{{\bm c},{\bm V}\}$, comprising $N$ visible biases $\bm c$ and $|E(G)|$ non-zero elements $V_{jk}$ for $j <k$ of a strictly upper triangular $N \times N$ matrix $\bm V$ of pairwise interactions between visible units. 

\begin{figure}[ht]
\begin{center}
\includegraphics[scale=0.5]{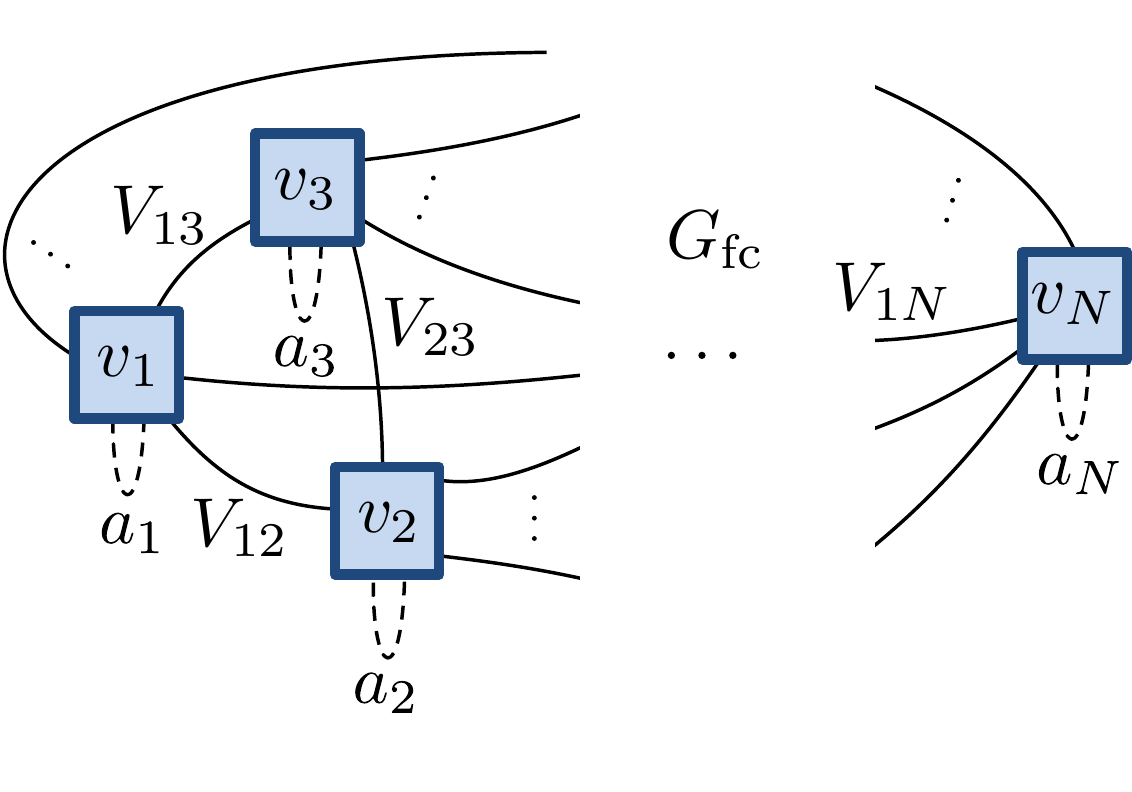}
\end{center}
\caption{The fully connected graph $G_{\rm fc}$ of interaction weights $\bm V$ between visible units in a generic Jastrow state. For completeness the visible biases $\bm a$ each unit are depicted here as self-loop edges with dotted arcs.}
\label{fig:jastrow_graph}
\end{figure}

In contrast to NQS the form of Jastrow states is severely limited by its two-body energy function for visible units in \eqr{eq:jastrow_energy}. As such Jastrow state amplitudes are simply the product of arbitrary two-spin states between all pairs connected by an edge in $G$ as
\begin{equation}
\psi_{\rm JS}({\bm v}) = \prod_{i=1}^{E(G)} \left(e^{c_{v_{{\tt s}_i}}v_{v_{{\tt t}_i}}+ c_{v_{{\tt t}_i}}v_{v_{{\tt t}_i}}}e^{V_{{\tt s}_i {\tt t}_i}v_{{\tt s}_i} v_{{\tt t}_i}}\right) = \prod_{i=1}^{E(G)} J^{(i)}_{v_{{\tt s}_i}v_{{\tt t}_i}}. \label{eq:jastrow_pairs}
\end{equation}
Jastrow states are therefore most expressive for a fully connected graph $G_{\rm fc}$, as shown in \fir{fig:jastrow_graph}. Given their similar forms of parameterisation it is natural to ask how to convert a Jastrow state into an NQS and what their hidden unit complexity is.

\subsection{Compact NQS for Jastrow states} \label{sec:jastrow_nqs}
A direct and commonly advocated~\cite{glasser18,nomura17,gao_dbm17,he_mdbm19} mapping of a Jastrow state into an NQS proceeds by mediating each pairwise Jastrow interaction $V_{jk}$ via an interaction with a hidden unit $h_i$. This is accomplished by solving the expansion
\begin{equation}
\exp(V_{jk}v_jv_k) \simeq \sum_{h_i=\pm 1}\exp(w_{ij}h_iv_j + w_{ik}h_iv_k), \label{eq:int_decomp}
\end{equation}
for each Jastrow interaction term in \eqr{eq:jastrow_energy}. There are many solutions for the hidden unit weights $w_{ij}$~\cite{carleo_dbm18,gao_dbm17,he_mdbm19,rrapaj2021}, such as
\begin{eqnarray}
w_{ij} &=& \frac{1}{2}\left[{\rm sech}^{-1}(e^{-V_{jk}}) + {\rm sech}^{-1}(e^{V_{jk}})\right], \quad {\rm and} \nonumber \\
w_{ik} &=& \frac{1}{2}\left[{\rm sech}^{-1}(e^{-V_{jk}}) - {\rm sech}^{-1}(e^{V_{jk}})\right], \label{eq:jastrow_sol1}
\end{eqnarray}
or a symmetric solution derived from the matrix square-root as
\begin{equation}
w_{ij} = w_{ik} = {\rm tanh}^{-1}\left[\sqrt{{\rm tanh}(V_{jk})}\right]. \label{eq:jastrow_sol2}
\end{equation} 
Inserting the decomposition \eqr{eq:int_decomp} into \eqr{eq:jastrow_amps} we arrive at a complex RBM formulation 
\begin{eqnarray}
\fl\qquad\quad\psi_{\rm JS}({\bm v})&\propto& \left[\prod_{j=1}^N \exp(c_jv_j)\right]\left[\prod_{i=1}^{|E(G)|} \sum_{h_i=\pm 1}\exp\left(w_{i{\tt s}_i}h_iv_{{\tt s}_i} + w_{i{\tt t}_i}h_iv_{{\tt t}_i}\right)\right], \\
&=& \sum_{{\bm h}}\exp\left(\sum_{j=1}^N c_jv_j + \sum_{i=1}^{|E(G)|}\left(w_{i{\tt s}_i}h_iv_{{\tt s}_i} + w_{i{\tt t}_i}h_iv_{{\tt t}_i}\right)\right).
\end{eqnarray}
Since this NQS requires $M = |E(G)|$ hidden units, each with the smallest receptive field of just two visible units~\cite{glasser18}, we will call it a sparse-Jastrow NQS. For the most expressive Jastrow state $G = G_{\rm fc}$ we thus need $M = \half N(N-1)$ hidden units in total. There is good reason to suspect that this sparse-Jastrow NQS significantly overestimates the hidden unit complexity of Jastrow states. In particular an NQS with $M=N$ hidden units, each having system-extensive connectivity, already has $N^2$ weights $\bm w$, which from a pure parameter count argument should be sufficient. Our first result indeed confirms this is the case: 

\begin{lemma}[A compact Jastrow state NQS] \label{lemma:jastrow}
Any Jastrow state for $N$ spins can be represented exactly by a compact NQS with $M=N$ system-extensive hidden units in which each hidden unit has a perfect correlation with a unique visible unit. 
\end{lemma}

\begin{proof}
The effective energy function in \eqr{eq:eff_energy} for a complex RBM can be written as
\begin{equation}
\fl \qquad    \mathcal{E}_{{\rm eff},{\bm \lambda}}({\bm v}) = \sum_{j=1}^N a_j v_j + {\rm log}\left[\prod_{i=1}^M \sum_{h_i=\pm 1}{\rm exp}\left(b_i h_i + \sum_{j=1}^N w_{ij}h_iv_j\right)\right]. \label{eq:eff_energy_jastrow}
\end{equation}
We take $M=N$ hidden units and a square interaction weight matrix $\bm w$ where the diagonal elements are $w_{ii} = \mathcal{S} \gg 1$. Focusing on the argument of the logarithm in \eqr{eq:eff_energy_jastrow}, we separate out the diagonal term $w_{ii}$, evaluate the sum over $h_i$ as $h_i=v_i$ and $h_i = -v_i$ and define $\theta_i({\bm v}) = b_i v_i + \sum_{j\neq i} w_{ij}v_iv_j$ to get
\begin{eqnarray*}
\fl \sum_{h_i=\pm 1}e^{b_i h_i + \sum_{j=1}^N w_{ij}h_iv_j} &=& \sum_{h_i=\pm v_i}e^{b_i h_i + \mathcal{S}h_iv_i + \sum_{j\neq i} w_{ij}h_iv_j} = e^{\mathcal{S}+\theta_i({\bm v})}\left(1 + e^{-2\mathcal{S}-2\theta_i({\bm v})}\right).
\end{eqnarray*}
Inserting this expansion for each hidden unit factor in \eqr{eq:eff_energy_jastrow} gives
\begin{equation}
\fl \qquad    \mathcal{E}_{{\rm eff},{\bm \lambda}}({\bm v}) = \sum_{j=1}^N a_jv_j + \sum_{i=1}^N \theta_i({\bm v}) + N\mathcal{S} + \sum_{i=1}^N{\rm log}\left[1+ e^{-2\mathcal{S}}e^{-2\theta_i({\bm v})}\right], \label{eq:jastrow_finite_en}
\end{equation}
where the last term generates multi-body interaction terms between visible units. However, in the limit $\mathcal{S} \rightarrow \infty$ these terms vanish giving a purely two-body effective energy function
\begin{equation}
\mathcal{E}_{{\rm eff},\infty}({\bm v}) = \lim_{\mathcal{S}\rightarrow\infty}\mathcal{E}_{{\rm eff},{\bm \lambda}}({\bm v}) = \sum_{i=1}^N (a_i+b_i)v_i + \sum_{i=1}^N\sum_{j \neq i} w_{ij}v_iv_j, 
\end{equation}
after dropping the irrelevant constant $N\mathcal{S}$. Thus, $\mathcal{E}_{{\rm eff},\infty}({\bm v})$ reproduces exactly the Jastrow amplitudes \eqr{eq:jastrow_amps} once ${\bm a} + {\bm b} = {\bm c}$ and ${\bm w} = \frac{1}{2}({\bm V} + {\bm V}^{\rm T}) + \lim_{\mathcal{S}\rightarrow \infty}\mathcal{S}\mathbbm{1}_{N\times N}$ with only $M=N$ hidden units. Like the sparse construction the hidden biases $\bm b$ are not needed and both have an identical number $N(N-1)$ of non-zero weights. However, in contrast to the sparse construction these weights have been concentrated into hidden units with the largest possible receptive field spanning the entire system. As such we will denote this compact NQS as an extensive-Jastrow NQS.
\end{proof}

The crucial step in the extensive-Jastrow NQS construction was the introduction of a diverging interaction weight $w_{ii} = \mathcal{S}$. This serves to perfectly correlate each hidden unit $h_i$ with one unique visible unit $v_i$ since
\begin{equation}
\lim_{\mathcal{S}\rightarrow\infty}\exp(\mathcal{S}h_iv_i) \propto \delta_{h_iv_i}, \label{eq:delta_coupling}
\end{equation}
once an overall scale factor is dropped. Owing to the presence of diverging weights $w_{ii}$ one might reasonably question whether the extensive-Jastrow NQS are numerically pathological compared to the sparse-Jastrow NQS. We now investigate this with the help of a non-trivial numerical example.

\subsection{Numerical example -- XXZ spin-chain} \label{sec:xxz}
Consider the XXZ model for a spin-$\half$ chain
\begin{equation*}
\hat{H}_{\rm XXZ} = \sum_{j=1}^N \left(\hat{X}_j\hat{X}_{j+1} + \hat{Y}_j\hat{Y}_{j+1} + \Delta \hat{Z}_j\hat{Z}_{j+1}\right), 
\end{equation*}
with an anisotropy $\Delta$ and periodic boundary conditions $N+1 \equiv 1$. This model displays three distinct ground state phases: i) a ferromagnetic phase ($\Delta \leq -1$), ii) a critical phase ($-1< \Delta \leq 1$) and iii) a gapped AF phase ($\Delta > 1$). The following translationally invariant antiferromagnetic spin-Jastrow state 
\begin{equation}
\fl\qquad\qquad \psi_{\rm CFT}({\bm v}) \propto \delta\left(\sum_{j=1}^N v_j\right)\, \prod_{j \in {\rm odd}} v_j \,\prod_{j > k} \left|\sin\left(\frac{\pi(j - k)}{N}\right)\right|^{\alpha v_j v_k}, \label{eq:cft_state}
\end{equation}
derived from a chiral boson conformal field theory (CFT), has been found to be a very good approximation of XXZ ground states for $\Delta > -1$ with an overlap (after restoring normalisation) $\mathcal{O} = |\braket{\psi_{\rm CFT}}{\psi_{\rm XXZ}}|^{2}$ in excess of 99\% for the majority of the critical region~\cite{cirac_mps10}. It comprises a $\delta(\sum_{j=1}^N v_j)$ contribution to enforce zero-$z$-magnetization constraint, $\prod_{j \in {\rm odd}} v_j$ to imprint the Marshall sign rule, and a product over positive-definite Jastrow factors controlled by the single positive parameter $\alpha$ related to the conformal dimension. The CFT state is exact for $\Delta = -1$ ($\alpha = 0$) and $\Delta = 0$ ($\alpha = \frac{1}{4}$), while its numerical minimisation for $-1 \leq \Delta \leq 1$ is well approximated by $\alpha =\frac{1}{2\pi}\arccos(-\Delta)$, as shown in \fir{fig:cft_plots}(a). For $\Delta = 1$ ($\alpha = \frac{1}{2}$) the CFT state is not exact and instead reduces to the well known Haldane-Shastry state \cite{haldane88,shastry88}. 

\begin{figure}[htb]
\begin{center}
\includegraphics[scale=0.5]{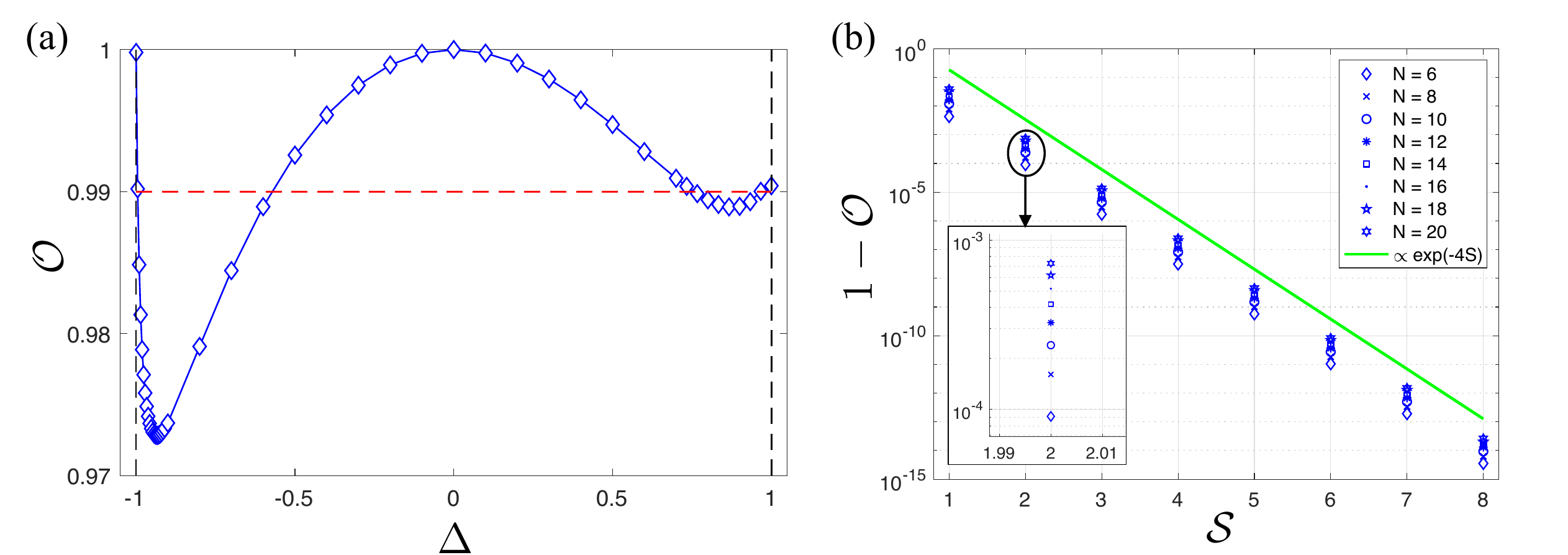}
\end{center}
\caption{(a) The overlap $\mathcal{O}$ with the exact XXZ ground state for $N=20$ spins as a function of the anisotropy $\Delta$ for $\psi_{\rm CFT}$ using $\alpha =\frac{1}{2\pi}\arccos(-\Delta)$. The line through the points is drawn to guide the eye, while the dashed horizontal line is the 99\% overlap. (b) The deviation of the overlap $1-\mathcal{O}$ between the $\Delta = 0$ exact Jastrow ground state and its softened NQS representation as a function of $\mathcal{S}$. The different symbols represent increasing system sizes $N = 6,\dots,20$, and the solid line is $\propto {\rm exp}(-4\mathcal{S})$ for reference. The panel shows a zoom of the $N$ dependence for $\mathcal{S} = 2$, which is illustrative of all values of $\mathcal{S}$ examined. The data and scripts used to create these plots in {\tt MATLAB} can be found in Ref.~\cite{jast_data}.}
\label{fig:cft_plots}
\end{figure}

The extensive-Jastrow NQS solution is rendered numerically benign by retaining a finite diagonal weight $w_{ii} = \mathcal{S}$. Specifically, given Jastrow interactions $\bm V$ we consider {\em softened} weights with $\mathcal{S}<10$. The $\Delta = 0$ XXZ ground state provides an ideal test case for this softened extensive Jastrow NQS. In \fir{fig:cft_plots}(b) we show the deviation of the overlap $1-\mathcal{O}$ between these states as a function of $\mathcal{S}$ for a sequence of increasing system sizes $N$. Two features are evident. First, the overlap converges to unity exponentially with increasing $\mathcal{S}$, consistent with the scaling of the multi-body terms in \eqr{eq:jastrow_finite_en}. Second, there is a weak decrease in $\mathcal{O}$ with increasing $N$, but this is easily controlled by the moderate values of $\mathcal{S}$ considered. Together this demonstrates the robust accuracy of the extensive-Jastrow NQS construction away from the formally exact $\mathcal{S} \rightarrow \infty$ limit.
 
\begin{figure}[htb]
\begin{center}
\includegraphics[scale=0.5]{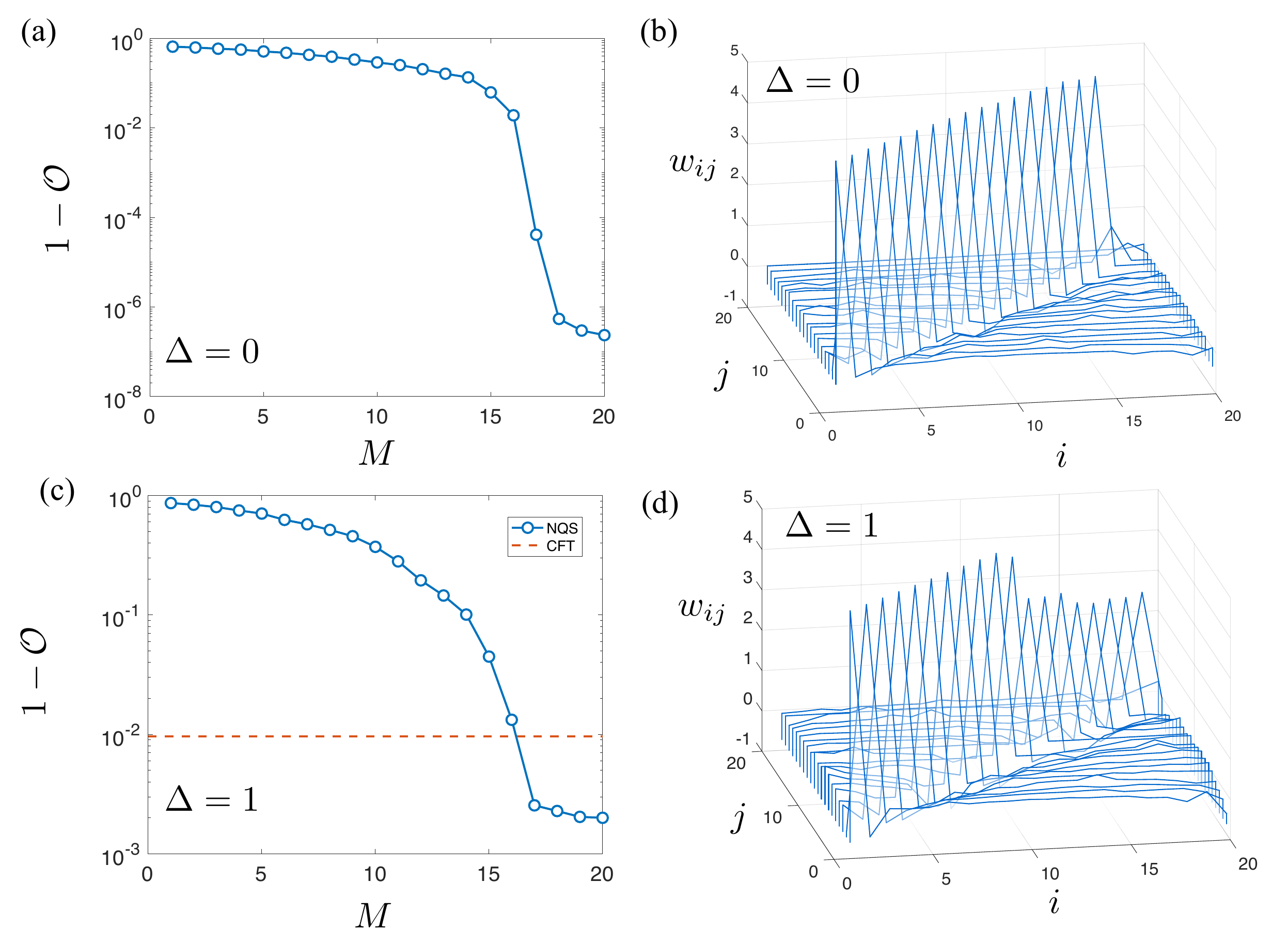}
\end{center}
\caption{(a) The deviation in the overlap $1-\mathcal{O}$ of the numerical NQS with the exact $N=20$ XXZ ground state for $\Delta = 0$ versus hidden unit number $M$. The lines drawn between points are to guide the eye. (b) The weights $w_{ij}$ between hidden units $i$ and visible units $j$ for the $M=N=20$ NQS $\Delta = 0$ solution. The interactions have been rearranged in order of decreasing maximum coupling strength. (c) The plot of $1-\mathcal{O}$ for $\Delta = 1$ where the XXZ ground state is not Jastrow. The dashed line is the value of $1-\mathcal{O}$ for the Jastrow CFT state. (d) The  weights $w_{ij}$ for the $M=N=20$ NQS $\Delta = 1$ solution. The data and scripts used to create these plots in {\tt MATLAB} can be found in Ref.~\cite{jast_data}.}
\label{fig:xxz_plots}
\end{figure}

To confirm the practical utility of the softened extensive-Jastrow NQS we performed VMC optimisation on an NQS for an $N=20$ XXZ chain. For numerical stability a parameter cap of $p_{\rm cap} = 5$ was applied. As is standard when applying VMC to the XXZ chain, the zero magnetisation constraint was enforced directly within the Monte Carlo sampling, removing the need to explicitly describe this with hidden units. Moreover, a gauge transform $\hat{\mathcal{G}} = \prod_{j \in {\rm odd}} \exp(-{\rm i}\pi \hat{Z}_j/2)$ was applied to $\hat{H}_{\rm XXZ}$ to give
\begin{equation*}
\hat{\mathcal{G}}\hat{H}_{\rm XXZ}\hat{\mathcal{G}}^\dagger = \sum_{j=1}^N \left(-\hat{X}_j\hat{X}_{j+1} - \hat{Y}_j\hat{Y}_{j+1} + \Delta \hat{Z}_j\hat{Z}_{j+1}\right), 
\end{equation*}
so the new Hamiltonian has exclusively non-positive off-diagonal elements in the $z$ basis, making it {\em stoquastic}. As a result its ground state is now guaranteed to have non-negative amplitudes in the $z$ basis~\cite{wang14} allowing us to restrict the NQS to real parameters. We minimised the NQS at $\Delta = 0$ for $1 \leq M \leq N$ using stochastic reconfiguration~\cite{sorella01,becca17} with system-extensive hidden units possessing randomly initialised weights and biases far below $p_{\rm cap}$. In Fig.~\ref{fig:xxz_plots}(a) we show $1-\mathcal{O}$ with the $\Delta = 0$ exact Jastrow ground state as a function of $M$. The deviation in the overlap displays a precipitous drop off of 4 orders of magnitude before plateauing at $M \geq 18$ above the softening limit seen in \fir{fig:cft_plots}(b) due to finite sampling fluctuations. We take the $1-\mathcal{O} \sim 10^{-7}$ as indicative of converging on an exact representation for $M<N$. Interestingly the weights $w_{ij}$ found by this numerical solution, reported in Fig.~\ref{fig:xxz_plots}(b), show that each hidden unit couples across the whole system and each interacts with a single unique visible unit with a weight that saturates $p_{\rm cap}$. These observations are consistent with the numerical optimisation ``learning" the softened extensive-Jastrow NQS structure with $\mathcal{S} = p_{\rm cap}$, demonstrating it is indeed a practical and stable NQS solution. The extensive-Jastrow NQS also explains earlier numerical observations in Ref.~\cite{glasser18} where an $M \sim O(N)$ scaling was found to describe exactly the Jastrow ground state of a 2D square lattice governed by the Laughlin state's parent spin Hamiltonian. 

The same optimisation scheme was performed for $\Delta=1$ where the Jastrow CFT state is not exact. Here $1-\mathcal{O}$ again plateaus for $M \geq 18$, but at a value orders of magnitude higher than at $\Delta = 0$, as shown in Fig.~\ref{fig:xxz_plots}(c). For the hidden unit numbers considered the NQS has not converged to the exact ground state. However, the NQS result does outperform the CFT state (Haldane-Shastry state) with $1-\mathcal{O}$ nearly 4 times smaller with the same number of hidden units. Examining the RBM interactions in Fig.~\ref{fig:xxz_plots}(d) shows that hidden units still favour optimisation into a Jastrow-like form, but with an increasingly softened interaction. This deviation is expected since NQS is a more expressive ansatz than Jastrow precisely due to the higher-order correlations introduced by the softened hidden units. 

The numerical results in Fig.~\ref{fig:xxz_plots}(a) suggest that even Jastrow states defined on a fully connected graph can be described with $M < N$ hidden units. While only a minor improvement from the $M=N$ extensive-Jastrow NQS introduced already we will find that understanding this behaviour opens the path for a considerable generalisation of the states that can be exactly captured by compact NQS.

\section{Tensor network formulation} \label{sec:tnt}
A powerful alternative formulation of NQS views them instead as a tensor network allowing for diagrammatic rewrites of their components~\cite{clark_cps18}. In this section we outline some key features of tensor network diagrams, introduce an NQS tensor network motivated from the RBM graph in \fir{fig:rbm}, and discuss diagrammatic observations about NQS tensor networks crucial for our main result.

\subsection{Tensor network diagrams} 
The exponentially many amplitudes $\Psi({\bm v})$ can be viewed as an order-$N$ tensor $\Psi_{v_1v_2\cdots v_N}$. Tensor network theory~\cite{schollwock11,verstraete08,cirac09,orus14} is a versatile way of handling this structureless tensor by decomposing it into many lower order tensors contracted together in a network. Here we will make repeated use of tensor network diagrams that form an important analytical tool in this approach. These represent generic tensors of any order as a shaded circle $\circ$ with protruding legs for each index its possesses. Given two order-3 tensors $A_{abc}$ and $B_{xyz}$ the contraction of them to form an order-4 tensor $C_{abxz} = \sum_\alpha A_{ab\alpha}B_{x\alpha z}$ is expressed as a graphical equation by joining the respective legs 
\begin{equation}
\includegraphics[scale=0.5,valign=c]{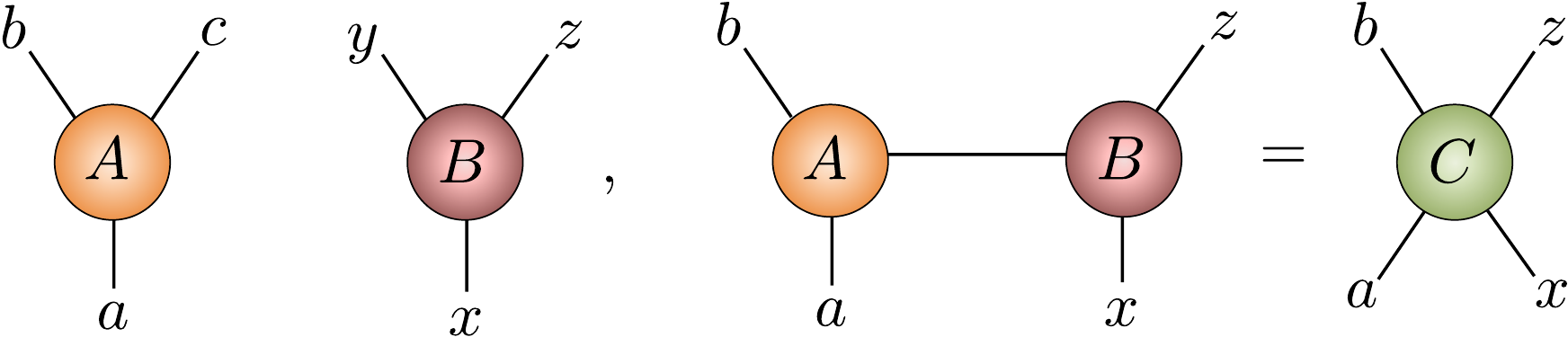} .
\end{equation}
For the most part we will only consider generic order-2 tensors, e.g. $A_{ab}$ which can be equivalently regarded as a matrix ${\bf A}$. We will use other shapes or shapes with symbols inside them to represent tensors with special structure. In particular, we will use a dot $\bullet$ to denote graphically the so-called COPY tensor. This is an essential building block for {\em sampleable} tensor networks in which the amplitudes $\Psi({\bm v})$ can be exactly and efficiently evaluated~\cite{clark_cps18}.

The COPY tensor~\cite{biamonte11,denny11} is the multi-index equivalent of the identity matrix. Specifically, for the case of three indices, the COPY tensor has elements
\begin{equation}
\delta_{ijk} = 
\left\{
\begin{array}{cc}
1, & i = j = k   \\
0, & {\rm otherwise}
\end{array}
\right., \label{eq:copy_tensor}
\end{equation}
which are zero unless all its indices are equal. It is represented diagrammatically as
\begin{equation}
\includegraphics[scale=0.5,valign=c]{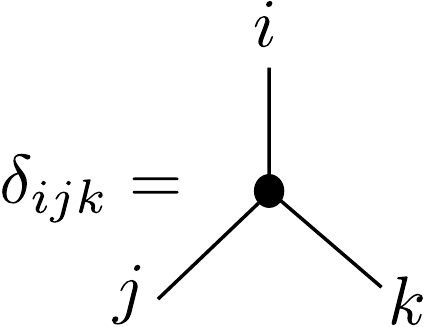} ,
\end{equation}
and it generalises straightforwardly to any number of indices. The name COPY tensor reflects that if any leg is contracted with a $z$ basis\footnote{A COPY tensor can be defined for any single spin basis by appropriately transforming each leg of this canonical one.} state $\ket{\uparrow} = ({1 \atop 0})$ and $\ket{\downarrow} = ({0 \atop 1})$ the same state gets copied to all the other legs, thereby factorising the tensor as
\begin{equation}
\includegraphics[scale=0.5,valign=c]{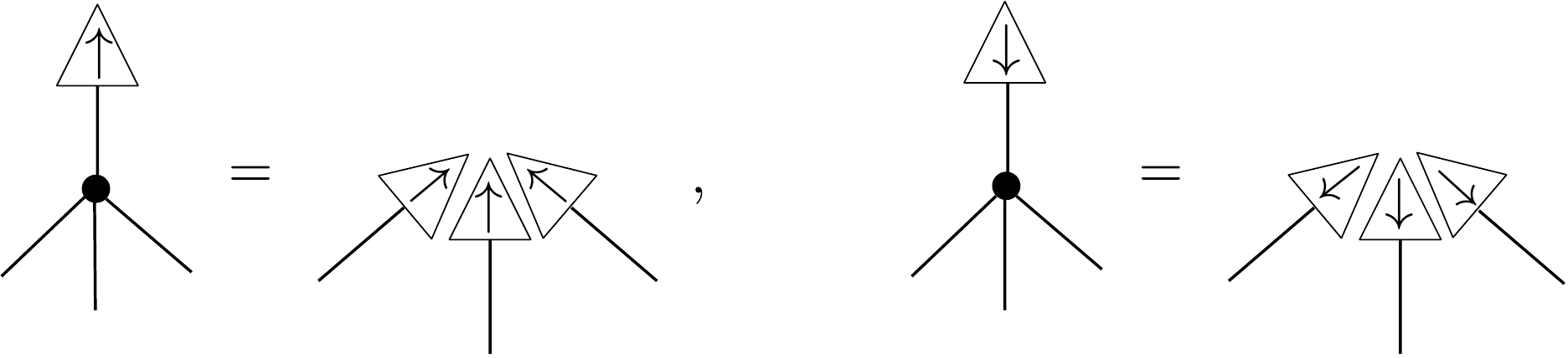} .
\end{equation}
Terminating any leg with an equal superposition $\ket{+} = \ket{\uparrow} + \ket{\downarrow} = ({ 1 \atop 1})$ removes the corresponding leg giving a COPY tensor with an order reduced by one
\begin{equation}
\includegraphics[scale=0.5,valign=c]{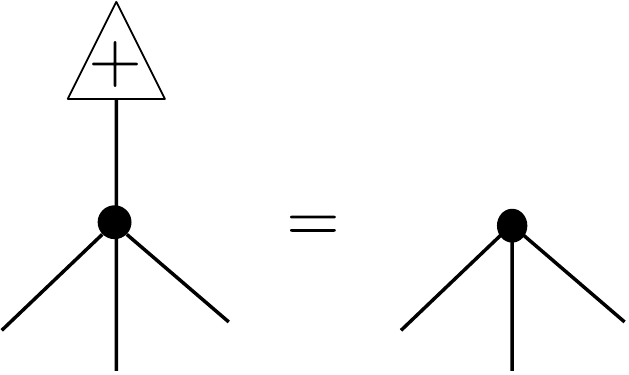} . \label{eq:copy_plus}
\end{equation}
The order-$N$ COPY tensor alone expands as the sum of two product terms
\begin{equation}
\fl\qquad\includegraphics[scale=0.5,valign=c]{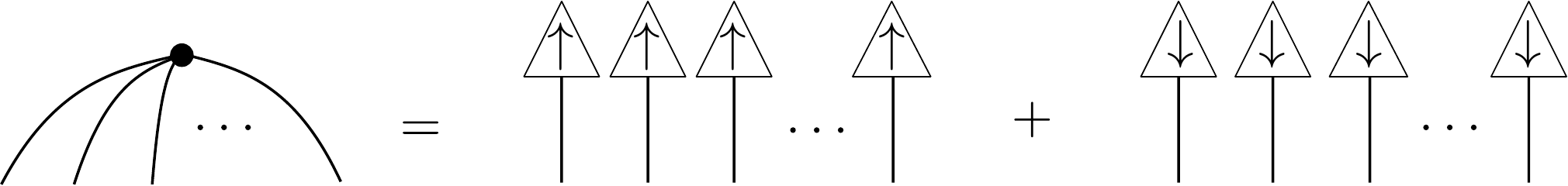} ,
\end{equation}
equivalent to the $z$ basis amplitudes of an $N$ spin ferromagnetic GHZ state. As we shall see shortly this basic tensor will form the skeleton of NQS.

A key property of COPY tensors is the so-called ``fusion" rule which allows COPY tensors having one or more legs contracted together to be amalgamated into one COPY tensor, e.g. as
\begin{equation}
\includegraphics[scale=0.5,valign=c]{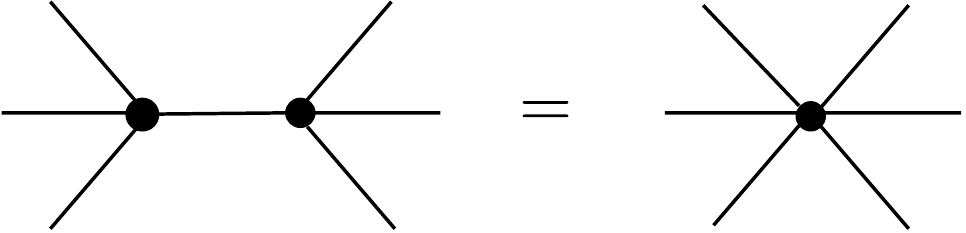} . \label{eq:copy_fusion}
\end{equation}
The rule also applies in reverse so a COPY tensor can be split up into an arbitrary network of connected COPY tensors with the same number of open legs.

A corollary of the fusion rule is that diagonal matrices can commute across the COPY tensor between any legs
\begin{equation}
\fl\qquad\includegraphics[scale=0.5,valign=c]{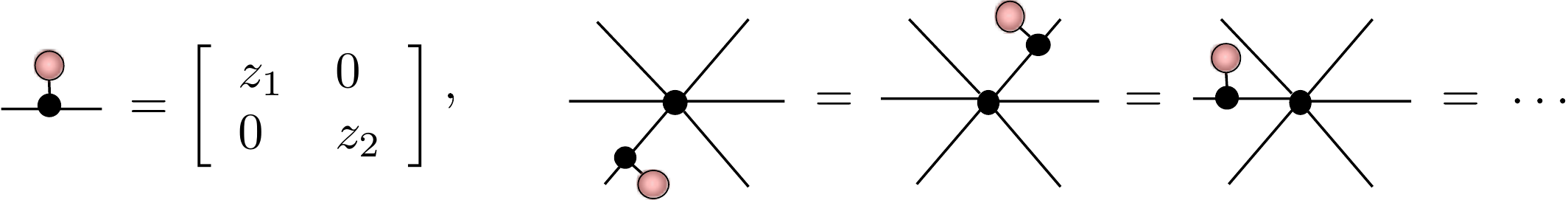} \quad. \label{eq:copy_diag}
\end{equation}
Similarly, the $\hat{X}$ component of anti-diagonal matrices distribute across legs as
\begin{equation}
\includegraphics[scale=0.5,valign=c]{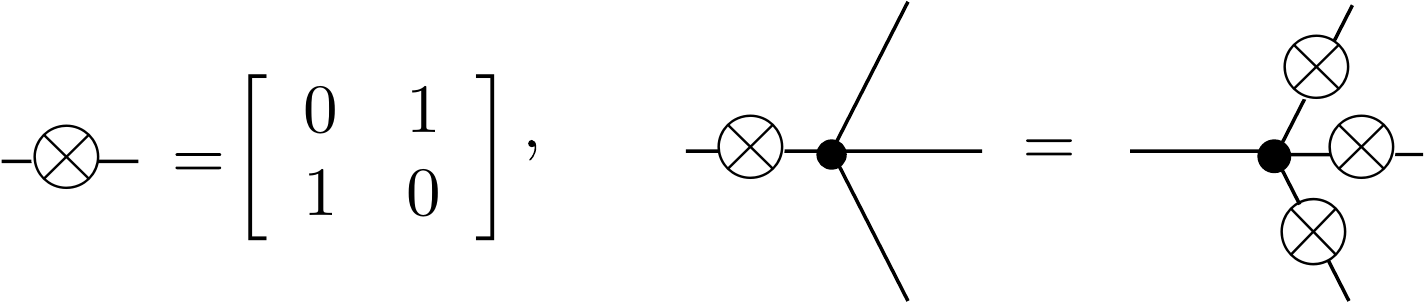} \quad. \label{eq:copy_adiag}
\end{equation}

\subsection{NQS tensor network}
The simple bipartite RBM graph in \fir{fig:rbm} readily motivates a corresponding tensor network representation~\cite{clark_cps18}. Specifically, each vertex is replaced by COPY tensor, each edge between the $i$th hidden and $j$th visible unit is replaced by a contraction with a generic $2 \times 2$ coupling matrix ${\bf C}^{(ij)}$, while vertices representing visible units have an additional open leg. Together this gives
\begin{equation}
\includegraphics[scale=0.5,valign=c]{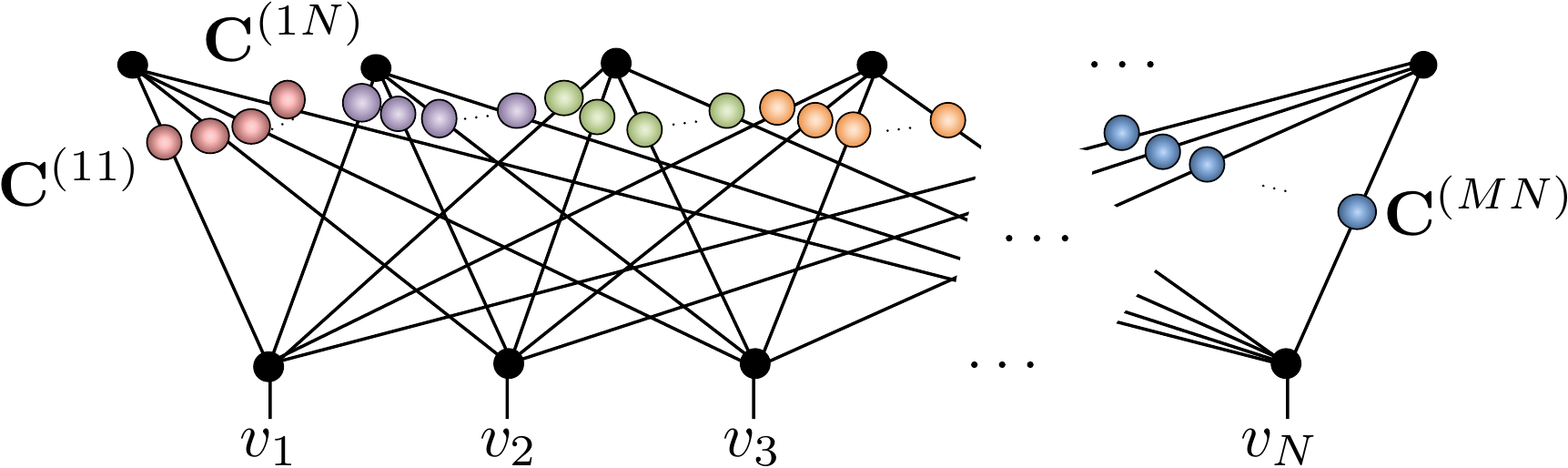} . \label{eq:nqs_tn}
\end{equation}
For this and forthcoming diagrams the colour of a tensor only denotes a convenient visual grouping, and each order-2 tensor is otherwise distinct. We will call any tensor network sharing this structure an {\em NQS tensor network}. Dissecting the network we see that in isolation each hidden unit is similar to a GHZ state, extensive over the whole system in general, but deformed locally by its coupling matrices as
\begin{equation}
\includegraphics[scale=0.5,valign=c]{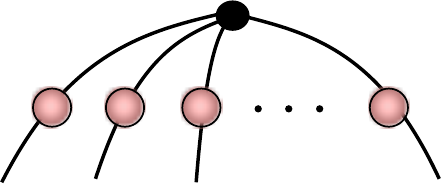} .
\end{equation}
The variational parameters provided by each hidden unit are therefore encoded by its set of $N$ coupling matrices $\Upsilon^{(i)} = \{{\bf C}^{(i1)},{\bf C}^{(i2)},\cdots,{\bf C}^{(iN)}\}$, explicitly tabulated as
\begin{equation}
\fl
\begin{array}{rcccc}
& \overbrace{+1 \quad -1}^{\textstyle v_1} & \overbrace{+1 \quad -1}^{\textstyle v_2} & \cdots & \overbrace{+1 \quad -1}^{\textstyle v_N}  \\
\\
h_i\,\left\{\begin{array}{c}
+1  \\
-1  
\end{array}\right. &
\left[
\begin{array}{cc}
C^{(i1)}_{++} & C^{(i1)}_{+-}   \\
C^{(i1)}_{-+} & C^{(i1)}_{--}
\end{array}
\right] & \left[
\begin{array}{cc}
C^{(i2)}_{++} & C^{(i2)}_{+-}   \\
C^{(i2)}_{-+} & C^{(i2)}_{--}
\end{array}
\right] & \cdots & \left[
\begin{array}{cc}
C^{(iN)}_{++} & C^{(iN)}_{+-}   \\
C^{(iN)}_{-+} & C^{(iN)}_{--}
\end{array}
\right]
\end{array}. \label{eq:cmat_tbl}
\end{equation}
As such the amplitudes of each hidden unit comprise two terms found by summing the product of coupling matrix elements selected by $\bm v$ along each row of \eqr{eq:cmat_tbl} as
\begin{equation}
\fl \qquad \Upsilon^{(i)}({\bm v}) = \sum_{h_i=\pm 1} \prod_{j=1}^N C^{(ij)}_{h_iv_j} = C^{(i1)}_{+v_1}C^{(i2)}_{+v_2}\cdots C^{(iN)}_{+v_N} \,+\, C^{(i1)}_{-v_1}C^{(i2)}_{-v_2}\cdots C^{(iN)}_{-v_N}. \label{eq:nqs_correlators}
\end{equation}
The amplitudes of the full NQS tensor network then follow as the product of each of these hidden unit correlators 
\begin{eqnarray}
\Psi_{\rm NQS}({\bm v}) &=& \prod_{i=1}^M \Upsilon^{(i)}({\bm v}), \label{eq:nqs_form}
\end{eqnarray}
which, like complex RBMs can be exactly and efficiently sampled. As discussed in Ref.~\cite{clark_cps18} coupling matrices can be a useful intuitive tool for unravelling the correlations and structures a given hidden unit imprints on the amplitudes of the overall NQS.

\subsection{Gauge freedom and RBM equivalence}
The NQS tensor network contains $4MN$ complex parameters, compared to $MN + M + N$ parameters for a complex RBM, suggesting these formulations are either not equivalent or that the NQS tensor network is over-parameterised. The latter is proven in the following result: 

\begin{lemma}[NQS tensor network RBM equivalence] \label{lemma:nqs2rbm}
An NQS tensor network for $N$ spin-$\half$ particles comprising $MN$ $2 \times 2$ coupling matrices ${\bf C}^{(ij)}$ is equivalent to the complex RBM formulation with $MN$ weights $\bm w$ and $N+M$ biases ${\bm a},{\bm b}$. 
\end{lemma}

The proof is presented in \ref{app:boltzmann}. The key step follows from \eqr{eq:copy_diag}, which demonstrates that NQS tensor networks possess a {\em gauge freedom}, analogous to that of MPS~\cite{schollwock11}, where (anti-)diagonal matrices can be shuffled between coupling matrices connected to the same COPY tensor. Given any coupling matrix can be decomposed into Boltzmann-like form as
\begin{eqnarray*}
\fl\qquad \left[
\begin{array}{cc}
C^{(ij)}_{++} & C^{(ij)}_{+-}   \\
C^{(ij)}_{-+} & C^{(ij)}_{--}
\end{array}
\right] &=& e^c\left[
\begin{array}{cc}
e^{\tilde{b}_{ij}}  & 0  \\
0  & e^{-\tilde{b}_{ij}}   
\end{array}
\right] \left[
\begin{array}{cc}
e^{w_{ij}} & e^{-w_{ij}}  \\
e^{-w_{ij}} & e^{w_{ij}}  
\end{array}
\right]\left[
\begin{array}{cc}
e^{\tilde{a}_{ij}}  & 0  \\
0  & e^{-\tilde{a}_{ij}}   
\end{array}
\right],
\end{eqnarray*}
complete equivalence to the complex RBM formulation in \eqr{rbm} and \eqr{rbm_energy} follows from reshuffling the diagonal matrices~\cite{chen18}. Armed with Lemma~\ref{lemma:nqs2rbm} we will proceed using the tensor network formulation of NQS and exploit the diagrammatic rewrites it affords in analytic constructions.

\subsection{Diagonal circuit unravelling}
An insightful alternative rewiring of a general NQS tensor network is made by elevating each hidden unit correlator to a diagonal operator in the $z$ basis
\begin{equation}
\hat{\Upsilon}^{(i)} = \sum_{{\bm v}} \Upsilon^{(i)}({\bm v})\ket{\bm v}\bra{\bm v}, \label{eq:correlator_operator}
\end{equation} 
so that the full NQS is constructed as 
\begin{equation}
\kets{\Psi_{\rm NQS}} = \prod_{i=1}^M \hat{\Upsilon}^{(i)}\ket{\Phi_+}, \label{eq:nqs_prepare}
\end{equation}
where $\ket{\Phi_+} = \ket{+}\ket{+}\cdots\ket{+} = \sum_{\bm v} \ket{\bm v}$ is the uniform superposition reference state. This is equivalent to diagrammatically unravelling each hidden unit in the NQS tensor network, using the COPY tensor fusion/splitting rule \eqr{eq:copy_fusion} and $\ket{+}$ state termination rule \eqr{eq:copy_plus}, into the form of a non-unitary preparation `circuit', as shown here for an NQS with $M=3$ system-extensive hidden unit:
\begin{equation}
\fl\qquad\includegraphics[scale=0.5,valign=c]{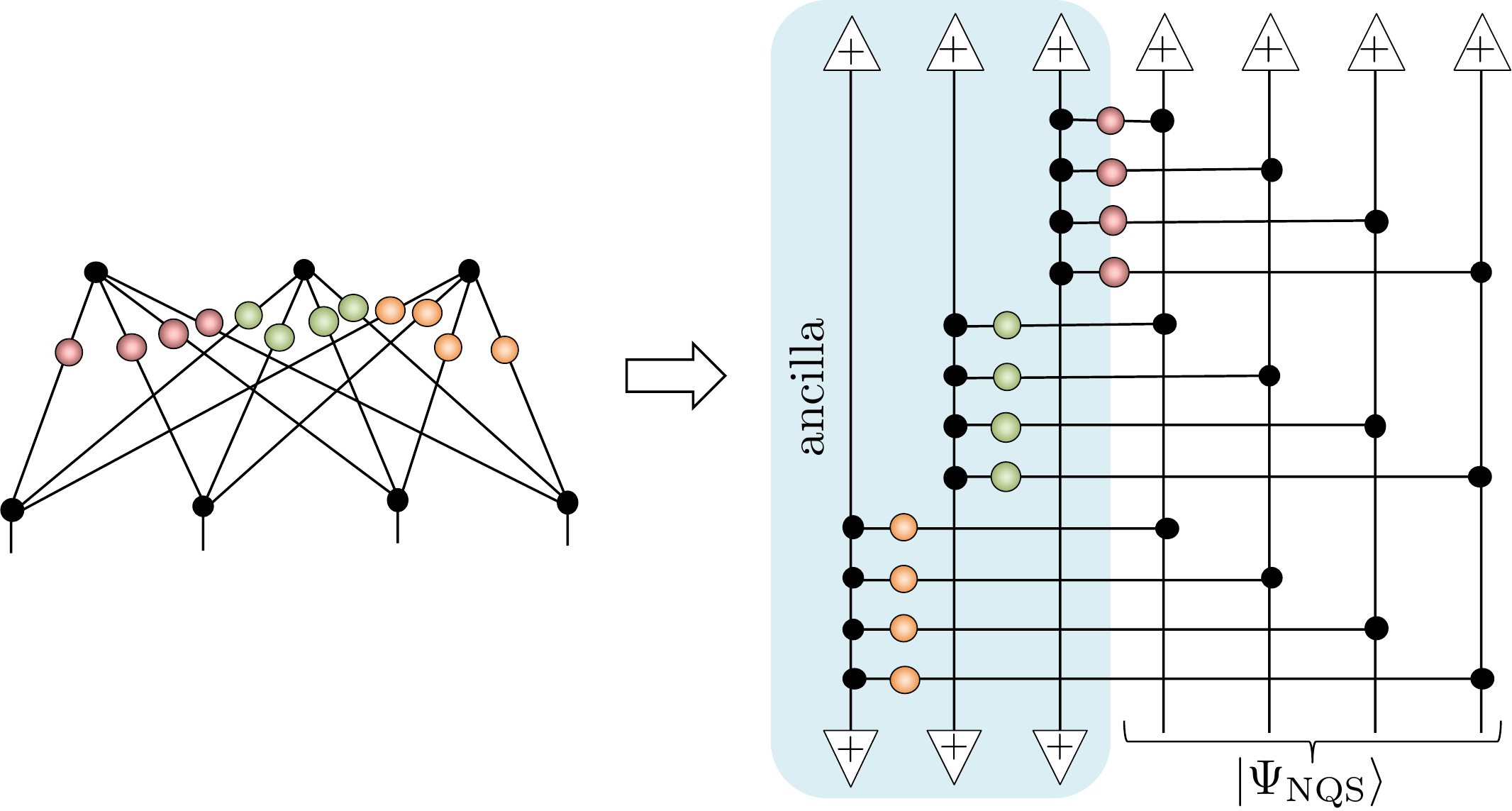} .
\end{equation} 
This circuit\footnote{Note that circuit diagrams in this paper are tensor networks diagrams. As such $\bullet$ denotes a COPY tensor which is similar but not identical to a control link found in standard quantum circuit diagrams~\cite{nielsen01}.} comprises entirely of diagonal two-spin operators and one ancilla spin per hidden unit that is initialised in $\ket{+}$ and projected out into $\ket{+}$ at the end. Since all operators commute their ordering in this circuit is irrelevant, reflecting the multiplicative structure of NQS. As each hidden unit involves a distinct ancilla the addition of more independently enhances the expressiveness of the ansatz.

\subsection{Applying single-spin operators to NQS}
Consider the following problem which is central to our main result. For an NQS initial state $\ket{\Psi_{\rm NQS}}$ what is the NQS representation of $\hat{Q}\ket{\Psi_{\rm NQS}}$ where a local operator $\hat{Q}$ is applied to a single spin of the physical system? In general this is non-trivial optimisation problem requiring both parametric and structural changes to the NQS~\cite{jonsson18,freitas18}. However, the tensor network formalism reveals two special cases where the update is simple: 

\begin{lemma}[Applying single spin operators to NQS] \label{lemma:single_spin}
The application of a single-spin operator $\hat{Q}$ to an NQS can be captured exactly by only changing the NQS parameters for two special cases: (i) An arbitrary operator $\hat{Q}$ applied to a visible unit that is univalent in the initial NQS; (ii) An operator $\hat{Q}$ that is (anti-)diagonal in the $z$ basis applied to any visible unit.
\end{lemma}

\begin{proof}
For case (i), where the visible unit $j$ is only one connected to a single hidden unit $i$, then $\hat{Q}$ is easily absorbed into the coupling matrix as ${\bf C}^{(ij)} \mapsto {\bf C}^{(ij)}{\bf Q}$, for example as:
\begin{equation}
\fl\qquad\quad\includegraphics[scale=0.5,valign=c]{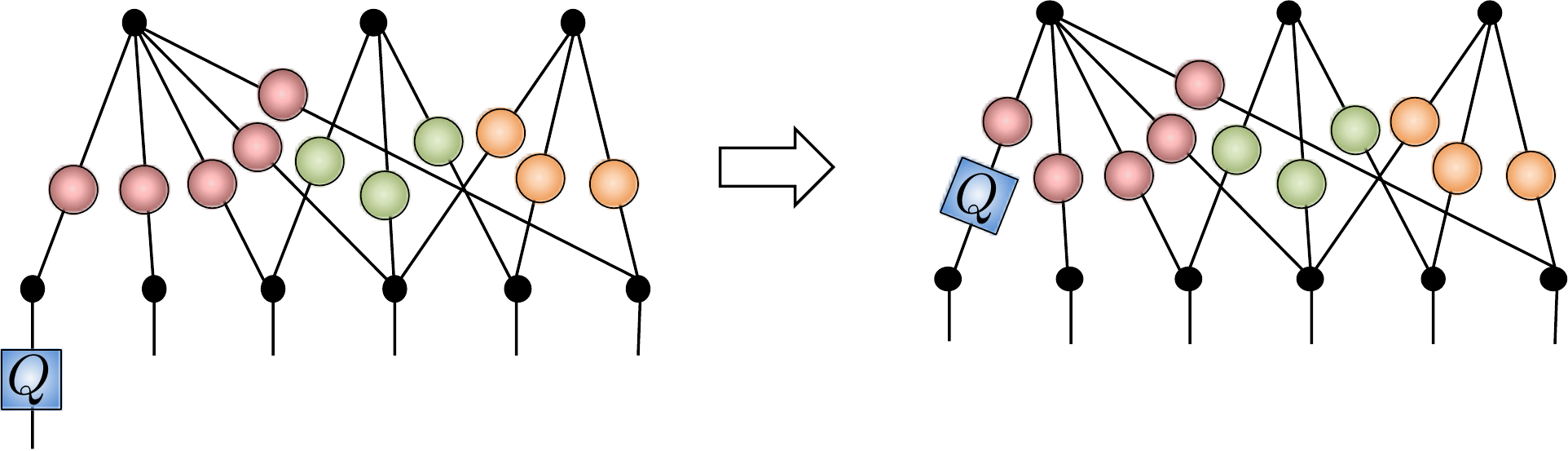}.
\end{equation}
For case (ii) $\hat{Q}$ can be applied to any spin since, via \eqr{eq:copy_diag} and \eqr{eq:copy_adiag}, it can be commuted past the visible unit's COPY tensor and absorbed into the NQS representation, for example as
\begin{equation}
\fl\qquad\quad\includegraphics[scale=0.5,valign=c]{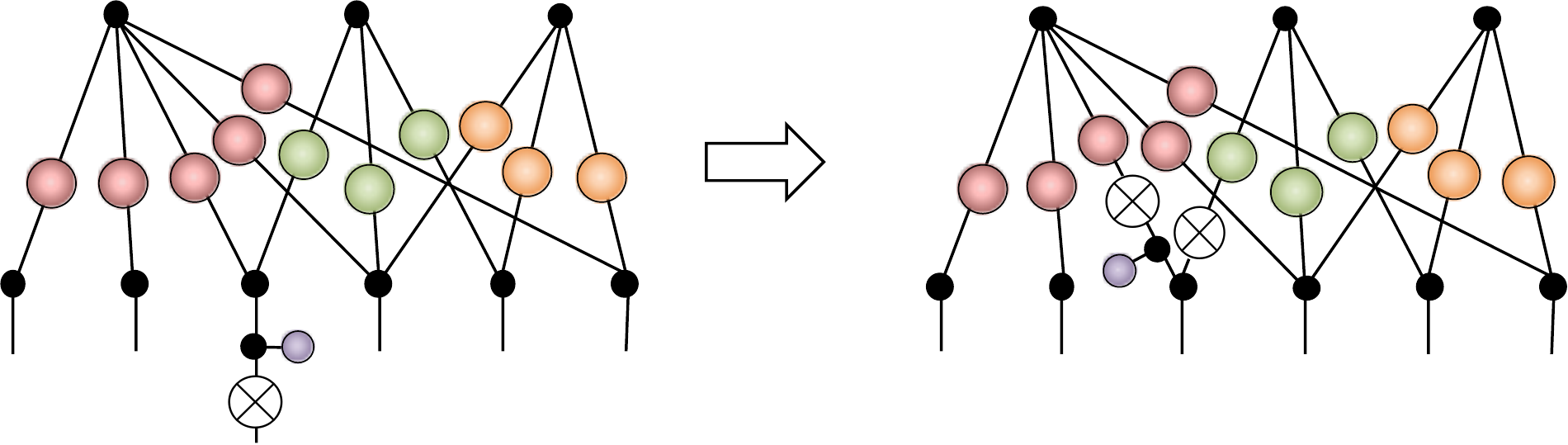} ,
\end{equation}
irrespective of the visible units valency in the NQS. Lemma~\ref{lemma:nqs2rbm} allows these parameter changes to be expressed in terms of complex RBM parameters, if required.
\end{proof} 

\section{Generalising compact Jastrow NQS} \label{sec:vmj}
In this section we apply the tensor network formalism to Jastrow states revealing a new type of compact NQS representation and a novel pathway for generalising it to wider classes of states.   

\begin{figure}[ht]
\begin{center}
\includegraphics[scale=0.5]{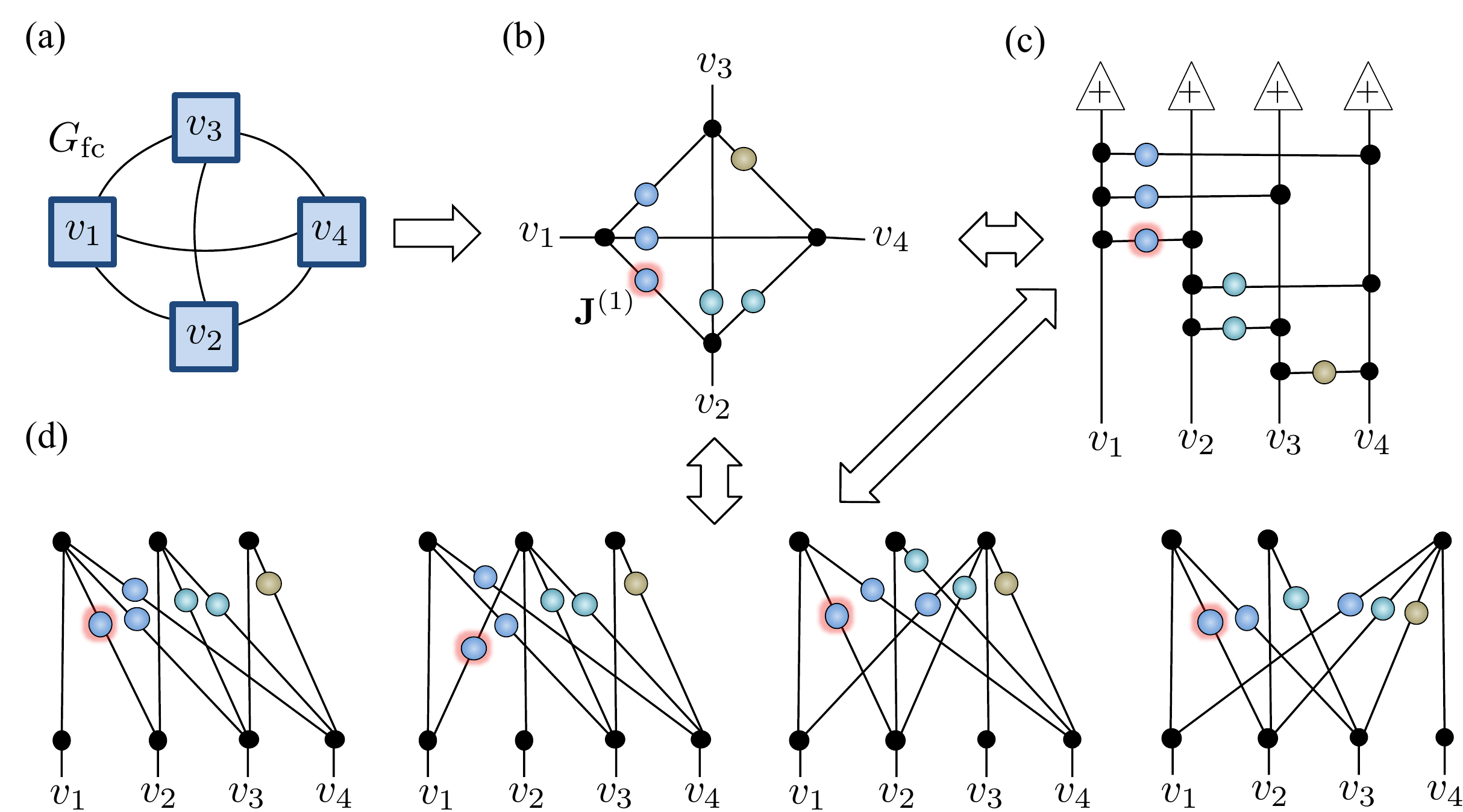}
\end{center}
\caption{(a) A fully connected graph $G_{\rm fc}$ for $N=4$ spins. (b) A CPS tensor network corresponding to $G_{\rm fc}$ in (a), with the first edge matrix ${\bf J}^{(1)}$ highlighted (and in each subsequent diagram as well to illustrate where it moves to). (c) A rewiring using \eqr{eq:copy_fusion} and \eqr{eq:copy_plus} of (b) into a circuit of diagonal two-spin operators. (d) Some possible rewirings of (b) and (c) into the geometry of an NQS tensor network. The examples drawn emphasise that any visible unit can be made univalent. All diagrams generalise straightforwardly for any $N$.}
\label{fig:jastrow}
\end{figure}

\subsection{Tensor networks for Jastrow states} \label{sec:jastrowtn}
Reformulating Jastrow states in terms of tensor networks provides several equivalent diagrammatic forms. Analogous to the RBM graph earlier, a Jastrow state defined over a graph $G$ leads to a sampleable tensor network constructed by replacing vertices with COPY tensors possessing an open leg and edges $i$ by contractions with $2 \times 2$ edge matrices ${\bf J}^{(i)}$. This produces a correlator product state (CPS) amplitude~\cite{changlani09}
\begin{equation}
\psi_{\rm JS}({\bm v}) = \prod_{i=1}^{E(G)} J^{(i)}_{v_{{\tt s}_i}v_{{\tt t}_i}}. \label{eq:jastrow_cps}
\end{equation}
Since the CPS tensor network is constructed from COPY tensors it possesses gauge freedom ensuring the representation in terms of ${\bf J}^{(i)}$ is completely equivalent to the Boltzmann parameterisation in \eqr{eq:jastrow_pairs} in terms of biases ${\bm c}$ and interactions $\bm V$. For the extremal case of $G_{\rm fc}$ in \fir{fig:jastrow}(a) the resulting CPS tensor network is shown in \fir{fig:jastrow}(b). The CPS tensor network can be straightforwardly unravelled into a circuit of diagonal two-spin operators, shown in \fir{fig:jastrow}(c). In contrast to the same unravelling of a general NQS there are no ancilla involved, consistent with the Jastrow energy function \eqr{eq:jastrow_energy} having interactions $\bm V$ only between visible units. Crucially the CPS and circuit networks can both be rewired in multiple ways to have the geometry of an NQS tensor network, a selection of which are shown in \fir{fig:jastrow}(d). 

Although specialised to the graph $G_{\rm fc}$ we can nonetheless glean several useful features about these new NQS tensor networks justifying why they should be considered the canonical Jastrow NQS form. First, they are compact possessing $M= N-1$ hidden units with decreasing coordination $N,N-1,\dots,2$. This explains why convergence to an exact numerical solution can appear with $M<N$ earlier in \secr{sec:xxz}. Second, they possess bare wires connecting some visible and hidden COPY tensors. From \eqr{eq:delta_coupling} these encode a diverging interaction enforcing a perfect correlation between those units, which in the tensor network language is equivalent to a coupling matrix ${\bf C}^{(ij)} = \mathbbm{1}_{2\times 2}$. Third, one of these new NQS forms, first in \fir{fig:jastrow}(d), has RBM weights ${\bm w} = {\bm V}' + \lim_{\mathcal{S}\rightarrow\infty} \mathcal{S}\mathbbm{1}_{N-1\times N}$, where ${\bm V}'$ denotes the $N$th (entirely zero) row removed from $\bm V$. Since $\bm V$ is a strictly upper-triangular matrix it also gives rise to the successively decreasing coordination of hidden units. Finally, they represent most direct and minimal translation of $G_{\rm fc}$ into an NQS since the $\half N(N-1)$ edge matrices ${\bf J}^{(i)}$ become the NQS coupling matrices. 

Contrast \fir{fig:jastrow}(d) to the tensor network versions of the Jastrow NQS representations introduced earlier in \secr{sec:jastrow_nqs}. For sparse-Jastrow NQS the tensor network possesses only order-2 hidden COPY tensors, as depicted here for a $N=4$ spin $G_{\rm fc}$ Jastrow state  
\begin{equation}
\includegraphics[scale=0.5,valign=c]{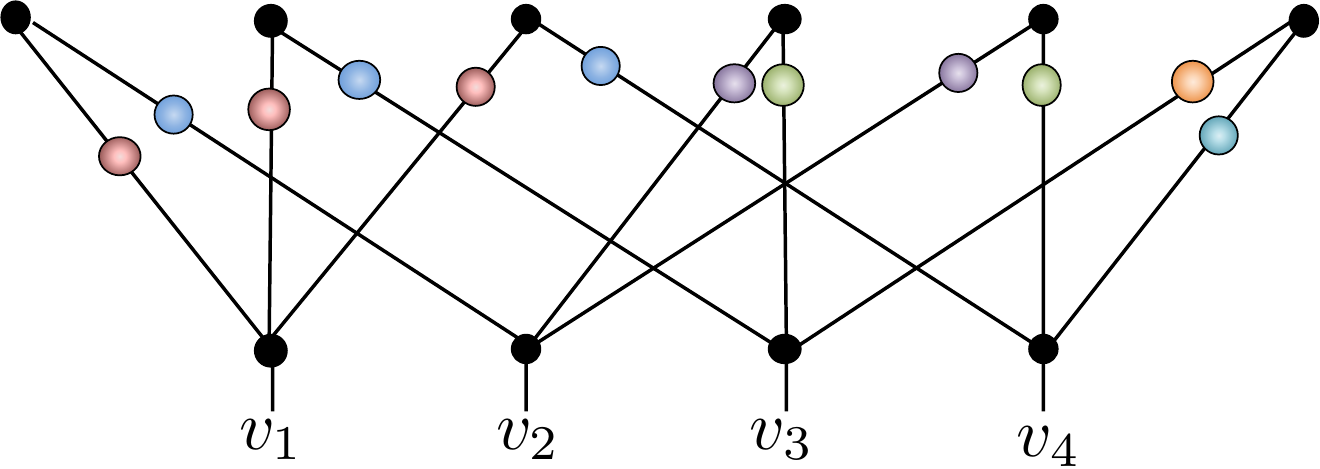} , \label{eq:jastrow_direct}
\end{equation}
while the extensive-Jastrow NQS has perfect correlations present between unique pairs of visible and hidden units, as shown here
\begin{equation}
\includegraphics[scale=0.5,valign=c]{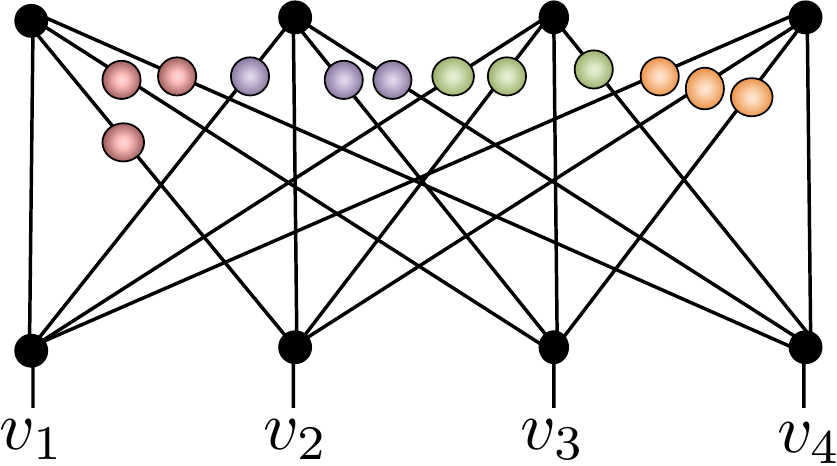} . \label{eq:compact_jastrow_nqs}
\end{equation}
Both involve $N(N-1)$ coupling matrices, but as expected can be diagrammatically rewired into the new forms shown in \fir{fig:jastrow}(d). However, crucially the new set of Jastrow NQS contain variants in which any single visible unit is univalent, as illustrated in \fir{fig:jastrow}(d). In light of Lemma~\ref{lemma:single_spin} this is an important property not present in either the sparse \eqr{eq:jastrow_direct} or extensive \eqr{eq:compact_jastrow_nqs} Jastrow NQS representations. 

By deleting couplings from the Jastrow NQS for $G_{\rm fc}$, equivalent to removing edges in the graph, we can readily obtain a Jastrow NQS defined over any graph $G$. This NQS will potentially possess more univalent visible units, dependent on the structure of the underlying graph $G$. For this reason we will now introduce some additional graph-theoretic concepts allowing us to formalise the Jastrow NQS construction for any graph $G$ and determine the freedom in its univalency.

\subsection{Graph theoretic tools} \label{sec:graph_theory}
We will exploit a number of basic concepts from graph theory~\cite{bondy2011} and illustrate them using an example graph $G$ from \eqr{eq:simple_graph} earlier
\begin{equation}
\includegraphics[scale=0.5,valign=c]{figures/graph.pdf} . \label{eq:graph_simple}
\end{equation}
First, the {\em neighbourhood} of a vertex $a \in \mathcal{V}$, denoted by $\mathcal{N}_a(G)$, is defined as the set of all vertices that are adjacent to vertex $a$, $\mathcal{N}_a(G) := \{b \in \mathcal{V} | (a,b) \in E(G)\}$. Second, a {\em leaf vertex} is a vertex with only a single edge incident on it, and we will denote the set of these in a graph as $\mathcal{L}(G)$. Third, an {\em independent set} $\mathcal{I}(G)$ is a set of vertices in which no pair shares an edge between them. Of the many such sets $\mathcal{I}(G)$ for a given graph those with the largest cardinality are maximum independent sets denoted as $\alpha(G)$. Fourth, a closely related concept is that of a {\em vertex cover}. A vertex cover $\mathcal{C}(G)$ is a set of vertices such that each edge of the graph is incident to at least one vertex in the vertex cover. Of the many sets $\mathcal{C}(G)$ those with the minimum cardinality form minimum vertex covers, denoted $\beta(G)$. We illustrate these four concepts here:
\begin{equation}
\fl\includegraphics[scale=0.5,valign=c]{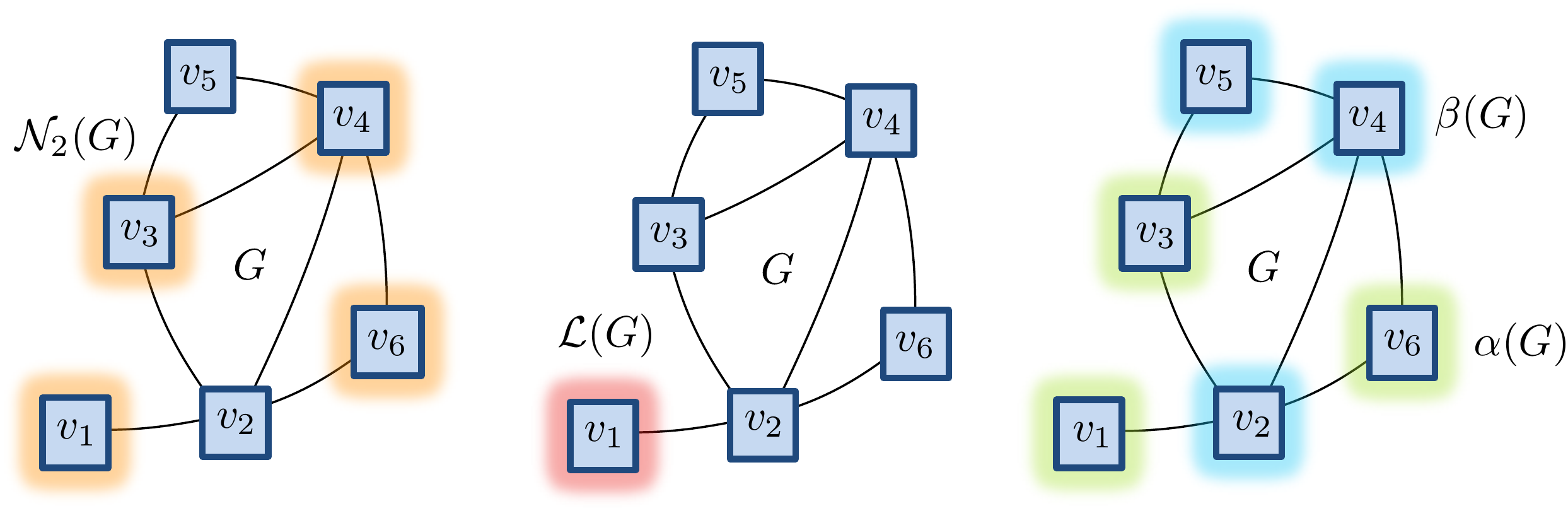} .
\end{equation}
A given maximum independent set has a corresponding minimum vertex cover that is its complement such that $\alpha(G) + \beta(G) = \mathcal{V}$, so finding one automatically gives the other. However, identifying either is an integer linear programming problem known to be NP-hard~\cite{karp72} in general.

\subsection{Vertex Modified Jastrow-NQS} \label{sec:ancilla_free}
Pulling together the observations from \secr{sec:jastrowtn} and tools from \secr{sec:graph_theory} we present here a general procedure {\tt graph2nqs} for constructing a Jastrow NQS tensor network directly from its graph $G$:
\begin{enumerate}
\item Begin the construction of the tensor network by inserting a COPY tensor with an open leg $-\!\bullet$ for each vertex in $G$, illustrated here for a simple example:
\begin{equation*}
\includegraphics[scale=0.5,valign=c]{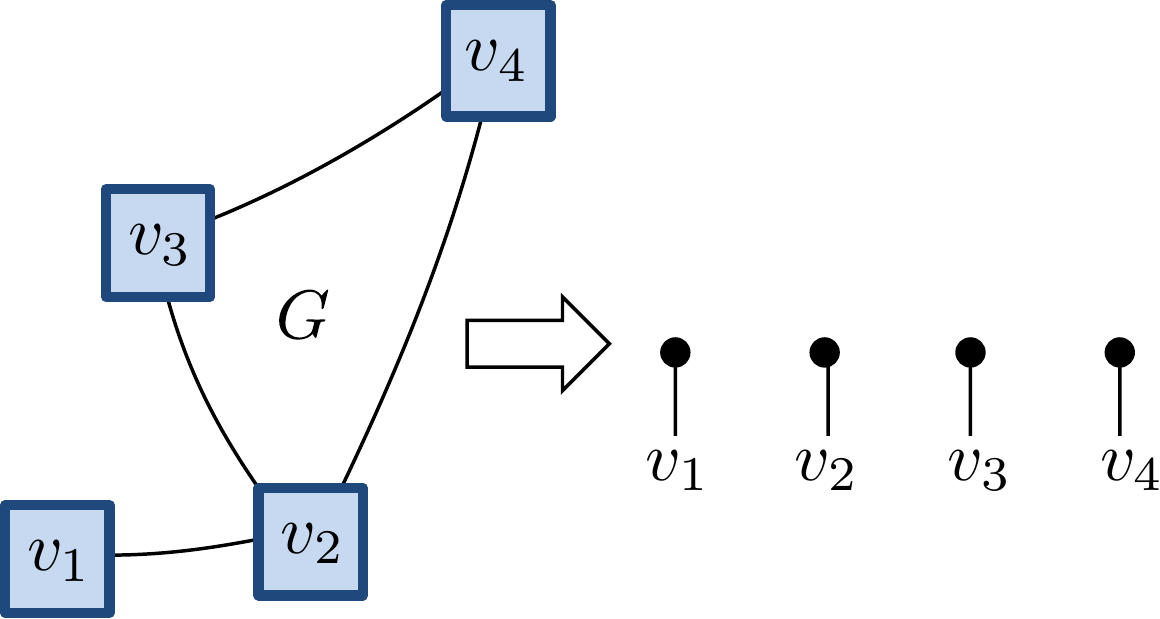} .
\end{equation*}
\item Pick any ordered vertex cover set $\mathcal{C}(G)$, and for each vertex in it split its corresponding COPY tensor into two as $-\!\!\!\bullet\!\!\!-\!\!\bullet$, with the first COPY tensor retaining the open leg, so it continues to represents the visible unit, and the second COPY tensor representing a perfectly correlated hidden unit:
\begin{equation*}
\includegraphics[scale=0.5,valign=c]{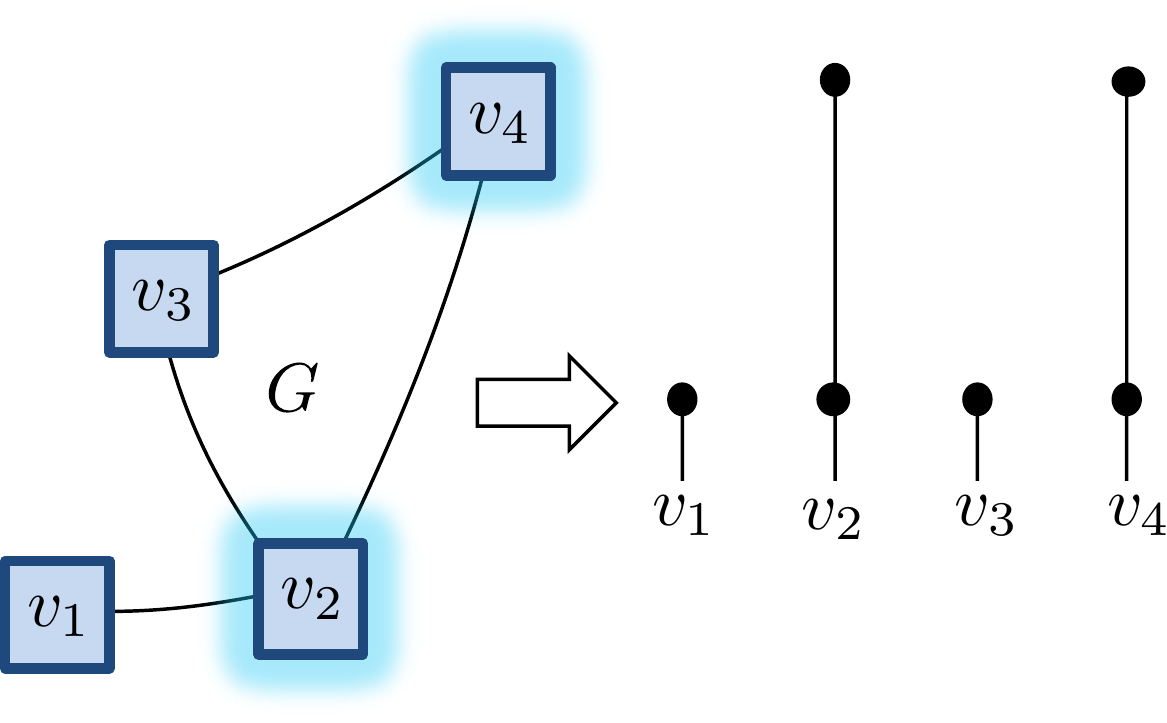} .
\end{equation*}
\item Proceed through $\mathcal{C}(G) = \{c_1,c_2,\dots\}$ in order. For vertex $c_j$ contract a separate order-2 tensor ${\bf J}^{(i)}$ $-\!\!\bigcirc\!\!-$ between $c_j$'s corresponding hidden COPY tensor and each visible COPY tensor corresponding to the vertices in
\begin{equation*}
l \in \mathcal{N}_j(G)/\left[\bigcup_{k=1}^{j-1} c_k\right],
\end{equation*}
neighbouring $c_j$, but with the previous members of $\mathcal{C}(G)$ removed:
\begin{equation*}
\fl\includegraphics[scale=0.5,valign=c]{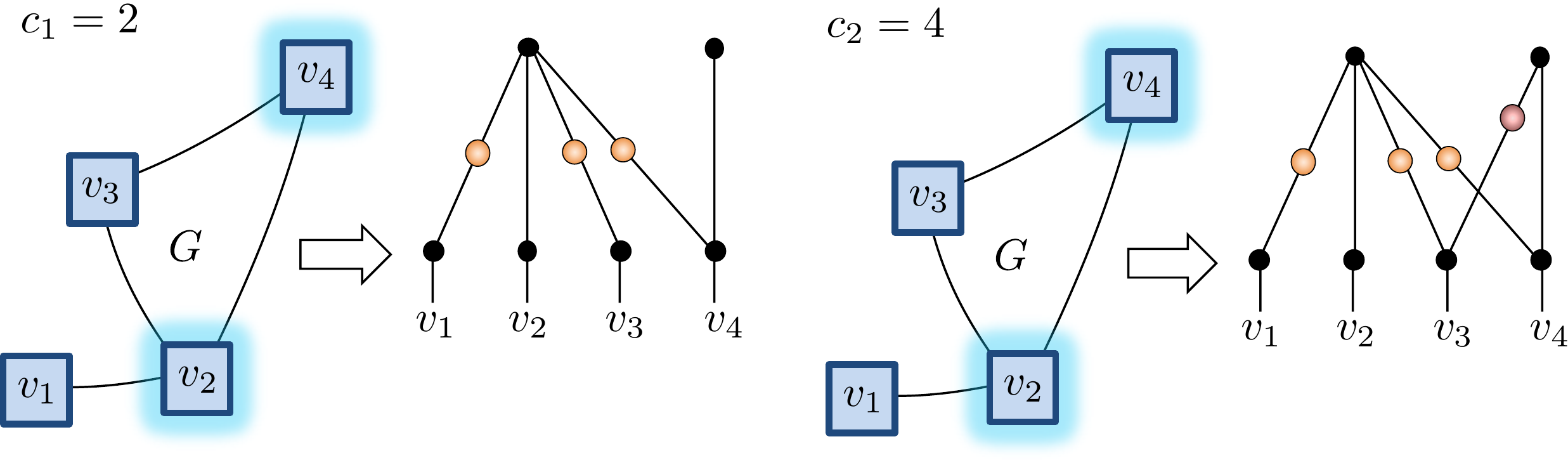} .
\end{equation*}
\end{enumerate}
Given we can choose any ordered vertex cover $\mathcal{C}(G)$ in {\tt graph2nqs} there is considerable freedom in our eventual Jastrow NQS. Here are some useful special cases:

\begin{lemma}[Minimum hidden unit Jastrow NQS] \label{lemma:jastrow_min}
A Jastrow state defined over $G$ can be expressed as a compact NQS tensor network comprising $M=|\beta(G)| \leq N-1$ hidden units.
\end{lemma}

\begin{proof}
Given {\tt graph2nqs} introduces perfectly correlated hidden units for each member of the vertex cover using $\beta(G)$ provides the minimum. The extremal case of the fully connected graph $G_{\rm fc}$, where $|\beta(G_{\rm fc})|=N-1$, is the maximum possible size for a vertex covering and so bounds the hidden unit complexity of any Jastrow state. 
\end{proof}

The set of leaf vertices $\mathcal{L}(G)$ in $G$ will always be univalent in an NQS. However, for {\tt graph2nqs} we have the further property that if the first members of the ordered vertex cover $\mathcal{C}(G)$ form an independent set $\mathcal{I}(G)$ then all the vertices $\mathcal{Q} = \mathcal{L}(G)\cup\mathcal{I}(G)$ will be univalent in the resulting Jastrow NQS. Using these observations we arrive at:

\begin{lemma}[Maximum univalency for a Jastrow NQS] \label{lemma:max}
Given a graph $G$ the maximum sized set of univalent visible units in its Jastrow NQS tensor network is $\mathcal{Q} = \mathcal{L}(G)\cup \alpha(G')$, where graph $G'$ is $G$ with leaf vertices pruned.
\end{lemma}

\begin{proof} 
Apply {\tt graph2nqs} with the following ordered vertex cover $\mathcal{C} = \{\alpha(G'),\beta(G'')\}$, where $G'$ is $G$ with leaf vertices pruned and $G''$ is $G$ with the vertices in $\alpha(G')$ removed. The visible units in $\mathcal{Q} = \mathcal{L}(G) \cup \alpha(G')$ will thus be univalent in the resulting Jastrow NQS tensor network, and is the maximum possible owing to $\alpha(G')$ being extremal. Note, the use of $\beta(G'')$ is to minimise the total hidden unit count, but is not strictly necessary. Any vertex cover $\mathcal{C}(G'')$ will suffice to complete $\mathcal{C}(G)$.  
\end{proof}

To illustrate these results we apply them here to a Jastrow state defined by the graph in \eqr{eq:graph_simple}. We find a maximum sized univalent set of vertices $\mathcal{Q}$ by applying the construction from Lemma~\ref{lemma:max} as
\begin{equation}
\fl\quad \includegraphics[scale=0.5,valign=c]{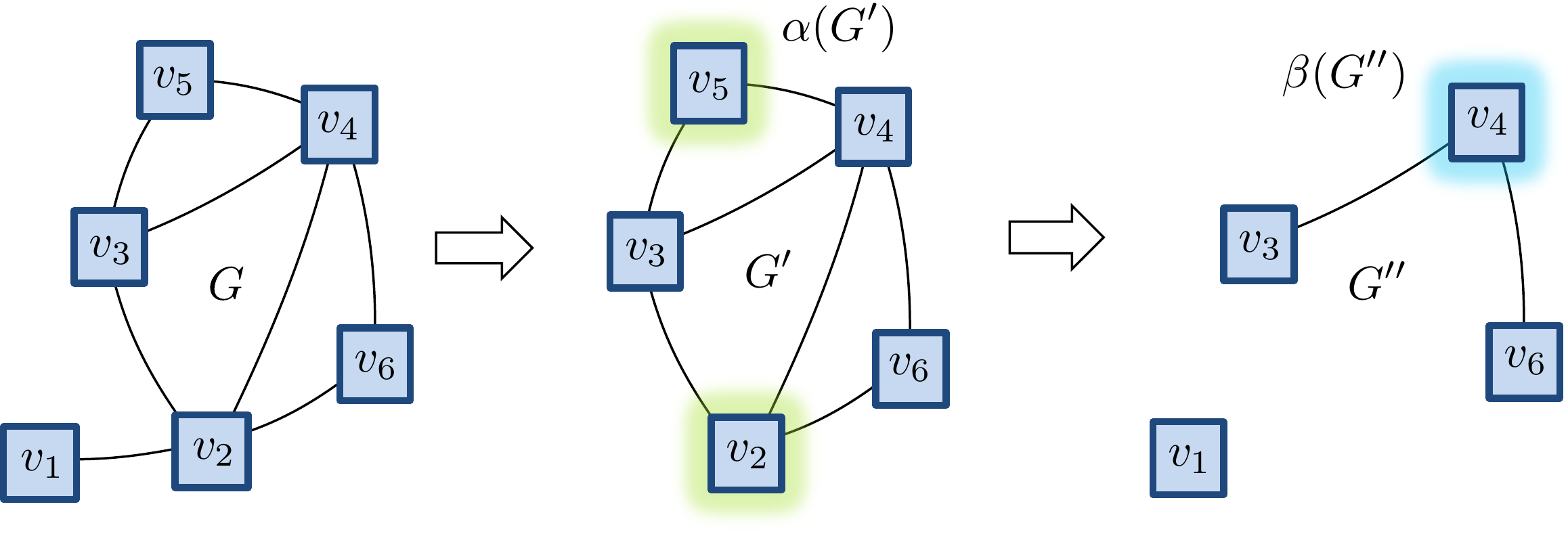} , \label{eq:graph_process}
\end{equation}
and forming an ordered vertex cover $\mathcal{C}(G) = \{2,5,4\}$. It turns out that $\mathcal{C}(G) = \beta(G)$ in this case, although it is not guaranteed to be so in general, and gives a Jastrow NQS
\begin{equation}
\includegraphics[scale=0.5,valign=c]{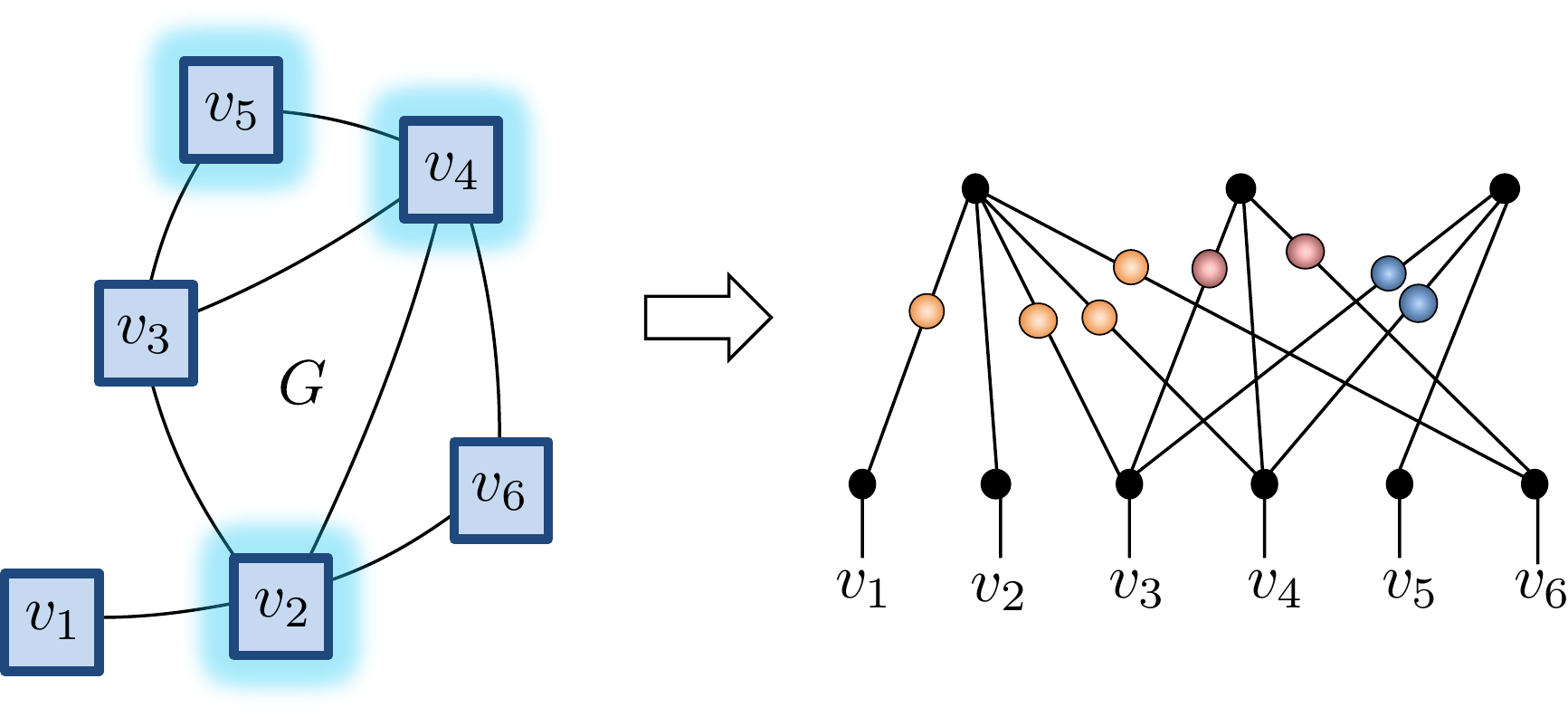} . \label{eq:graph_univalent}
\end{equation}
Since $|\mathcal{L}(G)| = 1$ and $|\alpha(G')| = 2$ in the example we arrive at a Jastrow NQS with univalent visible units $\mathcal{Q} = \mathcal{L}(G)\cup\alpha(G') = \{1,2,5\}$. 

The enlarged univalency of Jastrow NQS defined over graphs $G$ motivates the following generalisation:
\theoremstyle{definition}
\begin{definition}[VMJ-NQS]
A {\em Vertex Modified Jastrow NQS} (VMJ-NQS) is a Jastrow state defined over a graph $G$ which has arbitrary single-spin operators applied to any set of univalent vertices $\mathcal{Q}$ present in the Jastrow NQS as $\prod_{j \in \mathcal{Q}} \hat{Q}_j \ket{\Psi_{\rm JS}}$.
\end{definition}
By construction VMJ-NQS are compact since they possess the same $M \leq N-1$ hidden units as the underlying Jastrow state. Despite being a seemingly modest generalisation applying single-spin operators to $\mathcal{Q}$ has several implications. First, since operators are applied locally to single spins they cannot increase the entanglement content of the original Jastrow state $\ket{\Psi_{\rm JS}}$, and so we can view the $\hat{Q}_j$'s as fine-tuning. Second, like Jastrow states, VMJ states are equivalent to non-unitary preparation circuits with no ancilla
\begin{equation}
\fl\qquad\includegraphics[scale=0.5,valign=c]{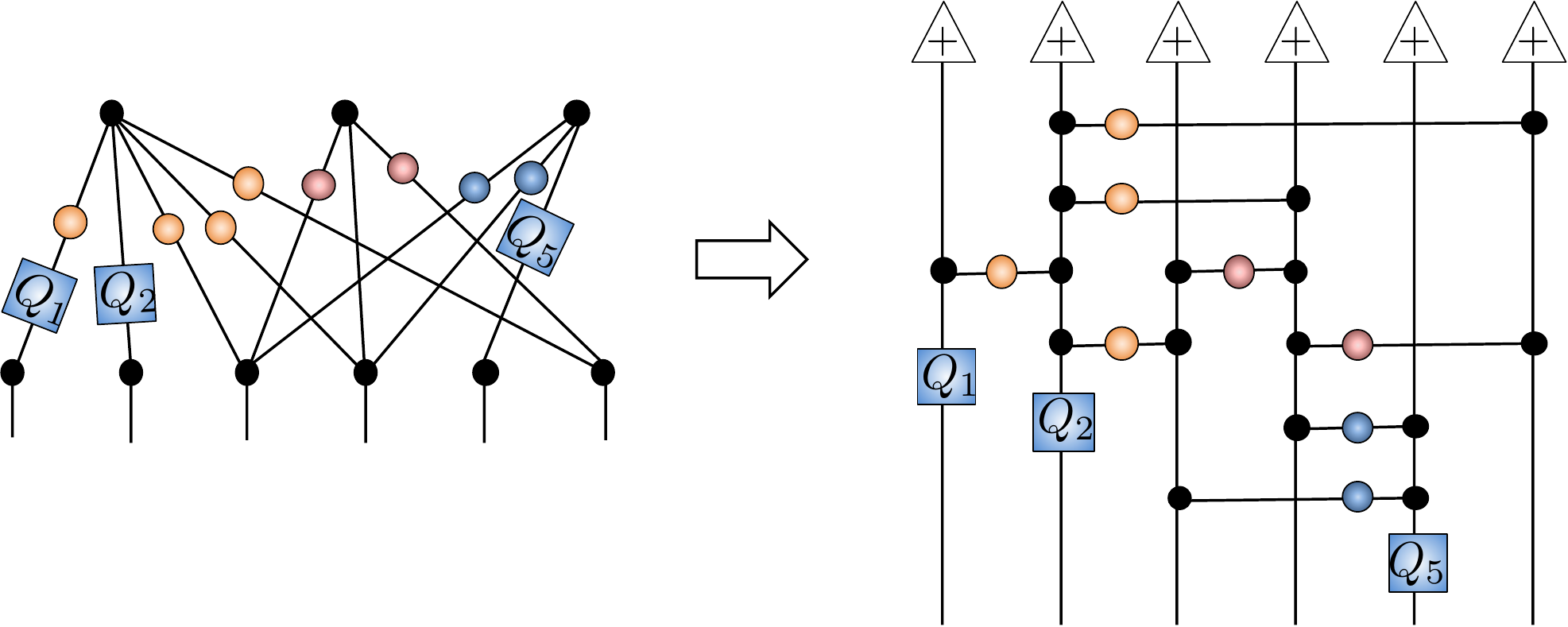} , \label{eq:graph_circuit}
\end{equation}
but possess fewer perfectly correlated visible and hidden units. Third, and most importantly, the presence of single-spin operators significantly enrich the nodal structure of VMJ states compared to Jastrow states. Specifically, for Jastrow states zero amplitudes can only be introduced by the zero elements in edge matrices ${\bf J}^{(i)}$, for example
\begin{equation*}
{\bf J}^{(i)} = 
\left[
\begin{array}{cc}
0  &  1  \\
1  &  0 
\end{array}
\right],
\end{equation*}
will kill {\em all} configurations $\ket{\bm v}$ in which the $v_{{\tt s}_i} = v_{{\tt t}_i}$. In contrast the application of a non-diagonal single-spin gate generates a superposition of Jastrow state amplitudes, conditioned on the state of the spin it was applied to, introducing the potential of vanishing amplitudes due to interference effects. We will now show that a consequence of this is that VMJ-NQS can capture a much wider class of quantum states, namely stabilizer states. 

\section{Graph states and stabilizer states} \label{sec:graph}
In this section we introduce graph states and stabilizer states along with an explicit procedure for constructing VMJ-NQS representations of them.

\subsection{Graph states}
The importance of graph states stems from their ability to possess volume scaling amounts of entanglement between subsystems~\cite{hein06} and that they form a resource for measurement-based quantum computation~\cite{raussendorf01,clark05}. A graph state $\ket{\Psi_{G}}$ is defined over $G$ for a set of spin-$\half$ particles all initialised in the state $\ket{+}$ as
\begin{equation}
\ket{\Psi_G} = \prod_{(j,k) \in E(G)} \hat{\rm C}^z_{jk} \ket{+}\ket{+}\cdots\ket{+}, \label{eq:graph_state}
\end{equation} 
in which a controlled-phase gate $\hat{\rm C}^z_{jk} = \mathbbm{1}_j\otimes\mathbbm{1}_k + \mathbbm{1}_j\otimes\hat{Z}_k$ is applied between any pair of spins $(j,k)$ if there is a corresponding edge in $E(G)$. The amplitudes $\Psi_G({\bm v})$ are thus non-zero and equal in magnitude for all $z$ basis states $\ket{\bm v}$, but possess an intricate sign structure imposed by the controlled-phase gates. Graph states correspond to a special class of spin-Jastrow state with a {\em unitary} preparation circuit (once normalisation is restored), a CPS tensor network with identical Hadamard-type edge matrices 
\begin{equation}
{\bf J}^{(jk)} = 
\left[
\begin{array}{cr}
1  & 1  \\
1 & -1  
\end{array}
\right],
\end{equation}
and an NQS representation with $M = |\beta(G)|$ hidden units~\cite{clark_cps18}. Graph states are also a special case of another class of states called stabilizer states. 

\subsection{Stabilizer states} \label{sec:stabilizers}
Stabilizer states are central to constructing quantum error-correction codes~\cite{gottesman97,nielsen01} and for describing the degenerate ground state manifolds of numerous interacting topological systems~\cite{wen07}. Their definition relies on the Pauli group for $N$ spin-1/2's $\mathfrak{P}_N$ which consists of $4 \times 4^N$ $N$-fold tensor product operators $\tau\, \hat{p}_1 \otimes \hat{p}_2 \otimes \cdots \otimes \hat{p}_N$, where $\tau \in \{\pm 1, \pm {\rm i}\}$ is an overall phase factor and each $\hat{p}_j \in \{\mathbbm{1}_j,\hat{X}_j,\hat{Y}_j,\hat{Z}_j\}$. The Clifford group $\mathfrak{C}_N$ for $N$ spins consists of all unitaries $\hat{U}$ whose action is to map under conjugation Pauli group elements among themselves, so $\hat{U}\mathfrak{P}_N\hat{U}^\dagger = \mathfrak{P}_N$. For a single spin the local Clifford group $\mathfrak{C}_1$, after disregarding a global phase, contains 24 unitaries including the Pauli gates as well as the Hadamard and phase gates~\cite{nielsen01}
\begin{equation}
\hat{\rm H} = \frac{1}{\sqrt{2}}
\left[
\begin{array}{cr}
 1 & 1   \\
 1 &  -1
\end{array}
\right], \quad {\rm and} \quad \hat{\rm S} = \left[
\begin{array}{cr}
 1 & 0   \\
 0 & {\rm i}
\end{array}
\right],
\end{equation}
respectively. The Clifford group $\mathfrak{C}_N$ can be generated by quantum circuits comprising the gates $\hat{\rm H}$ and $\hat{\rm S}$ along with the controlled-phase gate $\hat{\rm C}^z_{jk}$. 

The key idea of the stabilizer formalism is to represent a quantum state not by a vector of amplitudes but by a set of unitary operators that each ``stabilize" the state. Specifically, for stabilizer states these operators form an Abelian subgroup of $\mathfrak{P}_N$ defined by $N$ generators $\hat{T}_a$ that commute and are independent in the sense that removing a generator defines a smaller subgroup. A stabilizer state $\ket{\Psi_S}$ is then the unique eigenstate with +1 eigenvalue of each generator
\begin{equation}
\hat{T}_a\ket{\Psi_S} = \ket{\Psi_S}, \quad {\rm with} \quad a=1,2,\dots,N.
\end{equation}
A graph state $\ket{\Psi_G}$ is described by stabilizers
\begin{equation}
\hat{T}_a(G) = \hat{X}_a \prod_{b \in \mathcal{N}_a(G)} \hat{Z}_b. \label{eq:graph_stab}
\end{equation}
In contrast to graph states, the amplitudes $\Psi_S({\bm v})$ of stabilizer states can in general possess a nodal structure with zero amplitudes arising from parity constraints, while their non-zero amplitudes are equal magnitude with values $\pm 1, \pm {\rm i}$. Consequently stabilizer states can be written in a normal form
\begin{equation}
\ket{\Psi_S} = \sum_{\bm v} (-1)^{f({\bm q})}({\rm i})^{g({\bm q})}h({\bm q})\ket{\bm v}, \label{eq:stab_norm_form}
\end{equation}
in terms of Boolean-valued functions $f({\bm q}), g({\bm q}), h({\bm q})$ that are quadratic, linear and affine polynomials of the $z$ basis qubit labels ${\bm q} = \half(1 - {\bm v})$, respectively. 

Another powerful way to describe stabilizer states is to encode their generators into the rows of an $N \times 2N$ binary {\em check matrix}~\cite{aaronson04}
\begin{equation}
{\bf G} = \left[
\begin{array}{c|c}
{\bf X} & {\bf Z}
\end{array}
\right],
\end{equation}
partitioned into two $N \times N$ matrices {\bf X} and {\bf Z} with elements ${\rm x}_{aj}$ and ${\rm z}_{aj}$, respectively. The bits ${\rm x}_{aj}{\rm z}_{aj}$ determine the Pauli operator at the $j$th spin for the generator $\hat{T}_a$ as $00 \mapsto \mathbbm{1}_j$, $10 \mapsto \hat{X}_j$, $11 \mapsto \hat{Y}_j$ and $01 \mapsto \hat{Z}_j$, while elements ${\rm s}_a$ of an additional $N \times 1$ binary vector $\bf s$ specify the overall sign as $(-1)^{{\rm s}_a}$. The independence of the generators is equivalent to the rows of $\bf G$ being linearly independent, and they all mutually commute if and only if ${\bf G}{\bm \Lambda}{\bf G}^{\rm T} = 0$, where
\begin{equation}
{\bm \Lambda} = \left[
\begin{array}{c|c}
0 & \mathbbm{1}_{N\times N} \\
\mathbbm{1}_{N\times N} & 0
\end{array}
\right],
\end{equation}
defines a symplectic inner product for the rows. The check matrix is not unique since we are at liberty to swap rows in $\bf G$ and $\bf s$ simultaneously, corresponding to relabelling the generators. Similarly we can swap columns simultaneously within {\bf X} and {\bf Z} corresponding to relabelling the spins. Crucially we can also add rows modulo 2, corresponding to replacing a generator $\hat{T}_a$ with $\hat{T}_a\hat{T}_b$ when $a \neq b$, so long as we also update its sign $s_a$. Following \eqr{eq:graph_stab} the check matrix for a graph state has the form
\begin{equation}
{\bf G}_G = \left[
\begin{array}{c|c}
\mathbbm{1}_{N\times N} & {\bm \theta}
\end{array}
\right], \label{eq:graph_check}
\end{equation}
where $\bm \theta$ is its adjacency matrix of the graph $G$, and signs ${\bf s} = (0,0,\dots,0)$.

The usefulness of stabilizer states stems in a large part from the Gottesman-Knill theorem~\cite{gottesman97,nielsen01} which establishes that any Clifford quantum circuit $\hat{U} \in \mathfrak{C}_N$ acting on an $z$ basis state $\ket{\bm v}$ can be efficiently simulated classically. This follows since any $\ket{\bm v}$ is a stabilizer state with $\hat{T}_a = (-1)^{(1-v_a)/2}\hat{Z}_a$, so the state $\hat{U}\ket{\bm v}$ has stabilizers $\hat{T}'_a = \hat{U}\hat{T}_a\hat{U}^\dagger \in \mathfrak{P}_N$. The new stabilizers can be computed from elementary gates $\hat{\rm H}_j$, $\hat{\rm S}_j$ and $\hat{\rm C}^z_{jk}$ which induce simple updates on $\bf G$ and $\bf s$ for each generator $a=1,2,\dots,N$ as~\cite{aaronson04}:
\begin{itemize}
\item applying $\hat{\rm H}_j$, set ${\rm s}_a \mapsto {\rm s}_a \oplus ({\rm x}_{aj} \cdot {\rm z}_{aj})$ and then swap ${\rm x}_{aj}$ with ${\rm z}_{aj}$;
\item applying $\hat{\rm S}_j$, set ${\rm s}_a \mapsto {\rm s}_a \oplus ({\rm x}_{aj}\cdot {\rm z}_{aj})$ and then set ${\rm z}_{aj} \mapsto {\rm z}_{aj} \oplus {\rm x}_{aj}$;
\item applying $\hat{\rm C}^z_{jk}$, set ${\rm s}_a \mapsto {\rm s}_a \oplus ({\rm x}_{aj}\cdot {\rm z}_{ak})\cdot({\rm x}_{ak}\oplus{\rm z}_{aj}\oplus 1)\cdot({\rm x}_{ak}\oplus{\rm z}_{ak})$, and then set ${\rm z}_{aj} \mapsto {\rm z}_{aj} \oplus {\rm x}_{ak}$, ${\rm z}_{ak} \mapsto {\rm z}_{ak} \oplus {\rm x}_{aj}$.
\end{itemize}
We will exploit the first of these two updates in the following.

\subsection{VMJ-NQS for stabilizer states}
Significant efforts have been made to determine exact NQS representation of stabilizer states. Explicit constructions for the RBM parameterisation of stabilizer states have been devised based on iteratively reducing the normal form in \eqr{eq:stab_norm_form}~\cite{lu19}, and also by manipulating the check matrix into a canonical form~\cite{zhang_rbm2stbl18,zheng_stbl19,jia_surf19}. Analogous to Jastrow states, both these earlier constructions have a formal hidden unit complexity $M \sim O(N^2)$. Here we give a new approach exploiting that any stabilizer state $\ket{\Psi_S}$ is locally Clifford equivalent to a graph state $\ket{\Psi_G}$ as~\cite{schling01,vdn_cliff04}
\begin{equation}
\ket{\Psi_S} = \hat{u}_1\otimes\hat{u}_2\otimes \cdots \otimes\hat{u}_N \ket{\Psi_G}, \label{eq:stab_graph}
\end{equation}
for some $\hat{u}_a \in \mathfrak{C}_1$ with $a =1,2,\dots,N$. Given {\tt graph2nqs} provides a compact NQS for $\ket{\Psi_G}$ we can obtain a VMJ-NQS for $\ket{\Psi_S}$ from this so long as all the local Clifford gates $\hat{u}_a$ can be trivially absorbed via Lemma~\ref{lemma:single_spin}. Our main result is to show that this is indeed always possible:

\begin{theorem}[Stabilizer state NQS] \label{theorem:stab2nqs}
All stabilizer states have an exact VMJ-NQS representation, and so require no more than $M=N-1$ hidden units.
\end{theorem}

\begin{proof}
Similar to Ref.~\cite{zhang_rbm2stbl18} our proof exploits the mapping of the check matrix $\bf G$ of a stabilizer state $\ket{\Psi_{S}}$ into its canonical form. Specifically, by using row additions modulo 2 we perform Gaussian elimination from the top-left downwards, tracking the signs $\bf s$ and potentially swapping spins, to put the check matrix $\bf G$ into the form
\begin{equation}
{\bf G}^{(1)} = \left[
\begin{array}{cc|cc}
\mathbbm{1}_{r \times r} & {\bf A}  & {\bf B} & {\bf C} \\
0 & 0 & {\bf D} & {\bf E}
\end{array}
\right],
\end{equation}
where $r$ is the rank of $\bf X$, and $\bf A$ is an $r \times (N-r)$ matrix. Next, we perform Gaussian elimination from the bottom-right upwards for the last $N-r$ rows, again tracking the signs and doing any spin swaps necessary. Since all the rows of $\bf G$ were linearly independent originally the $(N-r) \times N$ submatrix $[{\bf D} ~ {\bf E}]$ must be full rank so we obtain the canonical form
\begin{equation}
{\bf G}^{(2)} = \left[
\begin{array}{cc|cc}
\mathbbm{1}_{r \times r} & {\bf A}  & \tilde{{\bf B}} & 0 \\
0 & 0 & \tilde{{\bf D}} & \mathbbm{1}_{(N-r) \times (N-r)}
\end{array}
\right].
\end{equation}
At this point we perform a sequence of local Clifford gates to render this canonical form into that of a graph state check matrix as in \eqr{eq:graph_check}. First, we perform $\hat{\rm H}$ gates on the last $N-r$ spins giving
\begin{equation}
{\bf G}^{(3)} = \left[
\begin{array}{cc|cc}
\mathbbm{1}_{r \times r} & 0  & \tilde{{\bf B}} & {\bf A} \\
0 &  \mathbbm{1}_{N-r \times N-r} & \tilde{{\bf D}} & 0
\end{array}
\right],
\end{equation}
making the new ${\bf X}^{(3)}$ matrix full rank and resulting in the stabilizer state now having non-zero amplitudes on all basis states $\ket{\bm v}$. The stabilizer commutativity condition ${\bf G}^{(3)}{\bm \Lambda}{\bf G}^{(3){\rm T}} = 0$ must still hold so
\begin{equation}
\left[
\begin{array}{cc}
\tilde{{\bf B}}\oplus \tilde{{\bf B}}^{\rm T}  & {\bf A} \oplus \tilde{{\bf D}}^{\rm T}\\
\tilde{{\bf D}} \oplus {\bf A}^{\rm T} & 0
\end{array}
\right] = \left[
\begin{array}{cc}
0 & 0 \\
0 & 0
\end{array}
\right],
\end{equation}
hence $\tilde{{\bf B}} = \tilde{{\bf B}}^{\rm T}$ and ${\bf A} = \tilde{{\bf D}}^{\rm T}$, meaning that the new ${\bf Z}^{(3)}$ matrix is symmetric. Second, ${\bf Z}^{(3)}$ is not yet an adjacency matrix due to $\tilde{\bf B}$ potentially having non-zero diagonal elements. These elements correspond to a $\hat{Y}_a$ operator in the corresponding generator $\hat{T}_a$ so we apply $\hat{\rm S}_a$ gates on any of the first $r$ spins where this is the case. This flips $\hat{Y}_a \mapsto -\hat{X}_a$ in $\hat{T}_a$ while leaving any $\hat{Z}_a$'s in other generators unchanged. Finally, we fix the pattern of signs $\bf s$ accumulated during these manipulations by applying $\hat{Z}_a$ to any generator which has ${\rm s}_a = 1$, inducing $\hat{X}_a \mapsto -\hat{X}_a$ flipping its sign, and overall making ${\bf s} = (0,0,\dots,0)$. We are left with a graph state check matrix of the form
\begin{equation}
{\bf G}_G = \left[
\begin{array}{cc|cc}
\mathbbm{1}_{r \times r} & 0  & \bar{\bf B} & {\bf A} \\
0 &  \mathbbm{1}_{N-r \times N-r} & {\bf A}^{\rm T} & 0
\end{array}
\right],
\end{equation}
where $\bar{\bf B}$ is $\tilde{\bf B}$ with its diagonal elements zeroed out. The matrix ${\bf Z}_G$ is now a valid adjacency matrix for a graph $G$. Having mapped $\ket{\Psi_S}$ to $\ket{\Psi_G}$ we now find the inverse for the sequence of local Clifford gates applied giving the equivalence in \eqr{eq:stab_graph}. Crucially the bottom zero corner of ${\bf Z}_G$ shows that the set of $N-r$ spins which had the non-diagonal $\hat{\rm H}$ gates applied to them form an independent set $\mathcal{I}_{\rm H}(G)$. The $\hat{\rm S}$ gates applied to the first $r$ are diagonal, as are $\hat{Z}$ gates. Consequently all gates applied are compatible with Lemma~\ref{lemma:single_spin} and can be trivially absorbed into the NQS generated by {\tt graph2nqs} using $\mathcal{I}_{\rm H}(G)$ at the start of its ordered vertex cover~\footnote{Note that the solution found is not unique. The procedure is dependent on the spin ordering and is biased towards identifying independent sets in the latter half of vertices. Other solutions can be accessed by permuting the spins (hence the columns of $\bf X$ and $\bf Z$) prior to the procedure and then inverse permuting both the columns and rows afterwards to restore the original ordering and retain the symmetry in the final $\bf Z$.}.
\end{proof}

\subsection{Analytic examples} \label{sec:stabilizer_examples}
To finish this work we now present some illustrative examples of Theorem~\ref{theorem:stab2nqs} applied to stabilizer states arising from quantum error correction codes and to a topological ground state. It is useful to introduce some specialised graphical representations of the Clifford gates $\hat{Z}$, $\hat{\rm S}^\dagger$ and $\hat{\rm H}$ as
\begin{equation}
\fl\qquad\includegraphics[scale=0.5,valign=c]{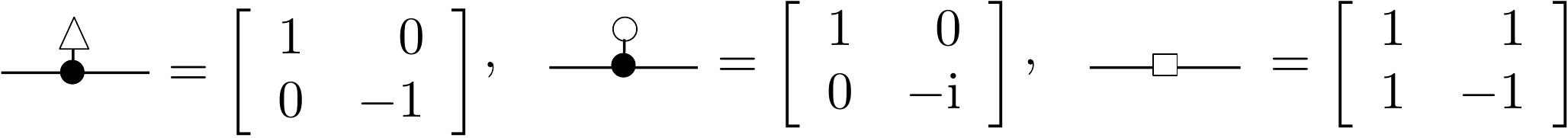} \quad , \label{eq:clifford_tensors}
\end{equation}
respectively, for the resulting NQS tensor network diagrams.

\subsubsection{Quantum error correcting code states:}
To provide a non-trivial illustration of our stabilizer state NQS construction we consider an encoded logical qubit state of the 7-qubit Steane code~\cite{steane96}. In this code the logical states in the $z$ basis, labelled as qubits $\ket{\bm q}$, are
\begin{eqnarray*}
\ket{0_{\rm L}} &=& \ket{0000000} + \ket{1010101} + \ket{0110011} +\ket{1100110} \\
&& +\ket{0001111}+\ket{1011010}+\ket{0111100}+\ket{1101001}, \\
\ket{1_{\rm L}} &=& \ket{1111111} + \ket{0101010} + \ket{1001100} +\ket{0011001} \\
&& +\ket{1110000}+\ket{0100101}+\ket{1000011}+\ket{0010110}.
\end{eqnarray*}
We add to the 6 generators defining this code space the operator $\prod_{j=1}^7 \hat{Y}_j$ so the unique stabilizer state is the superposition of logical qubit states $\ket{\Psi_{\rm steane}} = \ket{0_{\rm L}} - {\rm i}\ket{1_{\rm L}}$. Applying the procedure outlined in the proof of Theorem~\ref{theorem:stab2nqs} results in the following transformation of the stabilizer generators represented in this table
\begin{eqnarray*}
\fl\qquad\left[\begin{array}{ccccccc}
\mathbbm{1} & \mathbbm{1}  & \mathbbm{1} & X & X & X & X \\
\mathbbm{1} & X & X & \mathbbm{1} & \mathbbm{1} & X & X  \\
X & \mathbbm{1}  & X & \mathbbm{1} & X & \mathbbm{1} & X  \\
\mathbbm{1} & \mathbbm{1}  & \mathbbm{1} & Z & Z & Z & Z \\
\mathbbm{1} & Z  & Z & \mathbbm{1} & \mathbbm{1} & Z & Z \\
Z & \mathbbm{1} & Z & \mathbbm{1} & Z & \mathbbm{1} & Z \\
Y & Y  & Y & Y & Y & Y & Y \\
 &   &  &  &  &  &  \\
\hline
 &   &  &  &  &  &  \\
\mathbbm{1} & \mathbbm{1}  & \mathbbm{1} & \mathbbm{1} & \mathbbm{1} & \mathbbm{1} & \mathbbm{1} 
\end{array}\right] &\mapsto& 
\left[\begin{array}{ccccccc}
X & Z  & Z & \mathbbm{1} & \mathbbm{1} & Z & Z \\
Z & X & Z & \mathbbm{1} & Z & \mathbbm{1} & Z  \\
Z & Z  & X & \mathbbm{1} & Z & Z & \mathbbm{1}  \\
\mathbbm{1} & \mathbbm{1}  & \mathbbm{1} & X & Z & Z & Z \\
\mathbbm{1} & Z  & Z & Z & X & \mathbbm{1} & \mathbbm{1} \\
Z & \mathbbm{1} & Z & Z & \mathbbm{1} & X & \mathbbm{1} \\
Z & Z  & \mathbbm{1} & Z & \mathbbm{1} & \mathbbm{1} & X \\
  &   &  &  &  &  &  \\
\hline
 &   &  &  &  &  &  \\
{\rm S}^\dagger & {\rm S}^\dagger  & {\rm S}^\dagger & \mathbbm{1} & {\rm H} & {\rm H} & {\rm H}
\end{array}\right].
\end{eqnarray*}
The bottom row denotes the local Clifford gates to be applied to each qubit after the corresponding stabilizer state is constructed. The final form corresponds to a graph state defined with vertex Clifford operators shown in \fir{fig:qec_states}(a). As expected the non-diagonal $\hat{\rm H}$ gates are applied to the independent set $\mathcal{I}_G = \{5,6,7\}$ in this graph. We can then form a vertex cover $\mathcal{C}(G) = \mathcal{I}_G \cup \{1,3\}$ from which {\tt graph2nqs} then constructs as a VMJ-NQS with $M = 5$ hidden units, as shown in \fir{fig:qec_states}(b). In \fir{fig:qec_states}(c)-(d) we show some more examples of VMJ-NQS for stabilizer states arising from the 5-qubit $[[5,1,3]]$ code~\cite{laflamme96} and the 9-qubit Shor code~\cite{shor95}. 

\begin{figure}[ht]
\begin{center}
\includegraphics[scale=0.5]{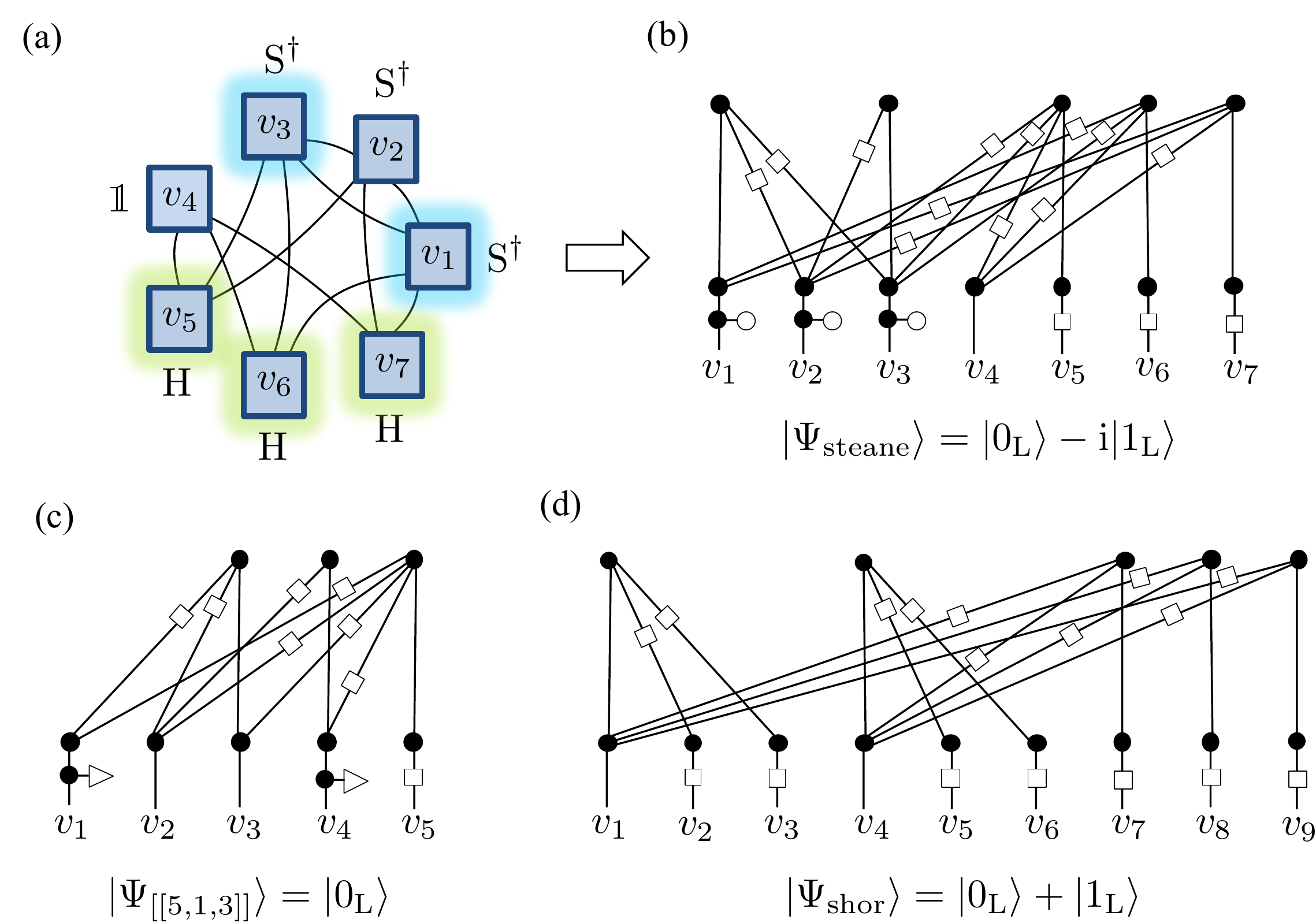}
\end{center}
\caption{(a) The graph defining a graph state locally Clifford equivalent to $\ket{\Psi_{\rm steane}} = \ket{0_{\rm L}} - {\rm i}\ket{1_{\rm L}}$, with the operators shown for each vertex. The independent set $\mathcal{I}_G =\{5,6,7\}$ is highlighted. (b) The VMJ-NQS constructed from the graph in (a) using the vertex cover
completed with the addition of the vertices $\{1,3\}$ also highlighted. The Clifford gates at each vertex can be trivially absorbed via Lemma~\ref{lemma:single_spin}. (c) The VMJ-NQS for the logical qubit state $\ket{\Psi_{\rm [[5,1,3]]}} = \ket{0_{\rm L}}$ of the 5-qubit $[[5,1,3]]$ quantum code. (d) The VMJ-NQS for the state $\ket{\Psi_{\rm shor}} = \ket{0_{\rm L}} + \ket{1_{\rm L}}$ constructed from the 9-qubit Shor code. Notice in this case the concatenated GHZ structure of the code is apparent.}
\label{fig:qec_states}
\end{figure}

\subsubsection{Toric code ground state}
A prominent example of a topological model is the Toric code Hamiltonian~\cite{kitaev03}
\begin{equation}
\hat{H}_{\rm toric} = -\sum_{s \in +} \prod_{j \in s} \hat{Z}_j -\sum_{p \in \square} \prod_{j \in p} \hat{X}_j,
\end{equation}
governing spins on a $2L \times L$ lattice located at the bonds of a 2D $L \times L$ square lattice with periodic boundary conditions. Here $+$ denotes the set of spins forming a star $s$ around a vertex, while $\square$ denotes the set of spins around the perimeter of a plaquette $p$ of the square lattice, as depicted in \fir{fig:toric_ham}(a). All the terms with $\hat{H}_{\rm toric}$ mutually commute and form a set of $2L^2$ generators in which any ground state is a simultaneous +1 eigenstate. However, since $\prod_{s \in +} \prod_{j \in s} \hat{Z}_j = \mathbbm{1}$ and $\prod_{p \in \square} \prod_{j \in p} \hat{X}_j = \mathbbm{1}$ this set contains only $2(L^2-1)$ independent generators and so defines a 4-dimensional degenerate ground state code subspace. 

\begin{figure}[ht]
\begin{center}
\includegraphics[scale=0.5]{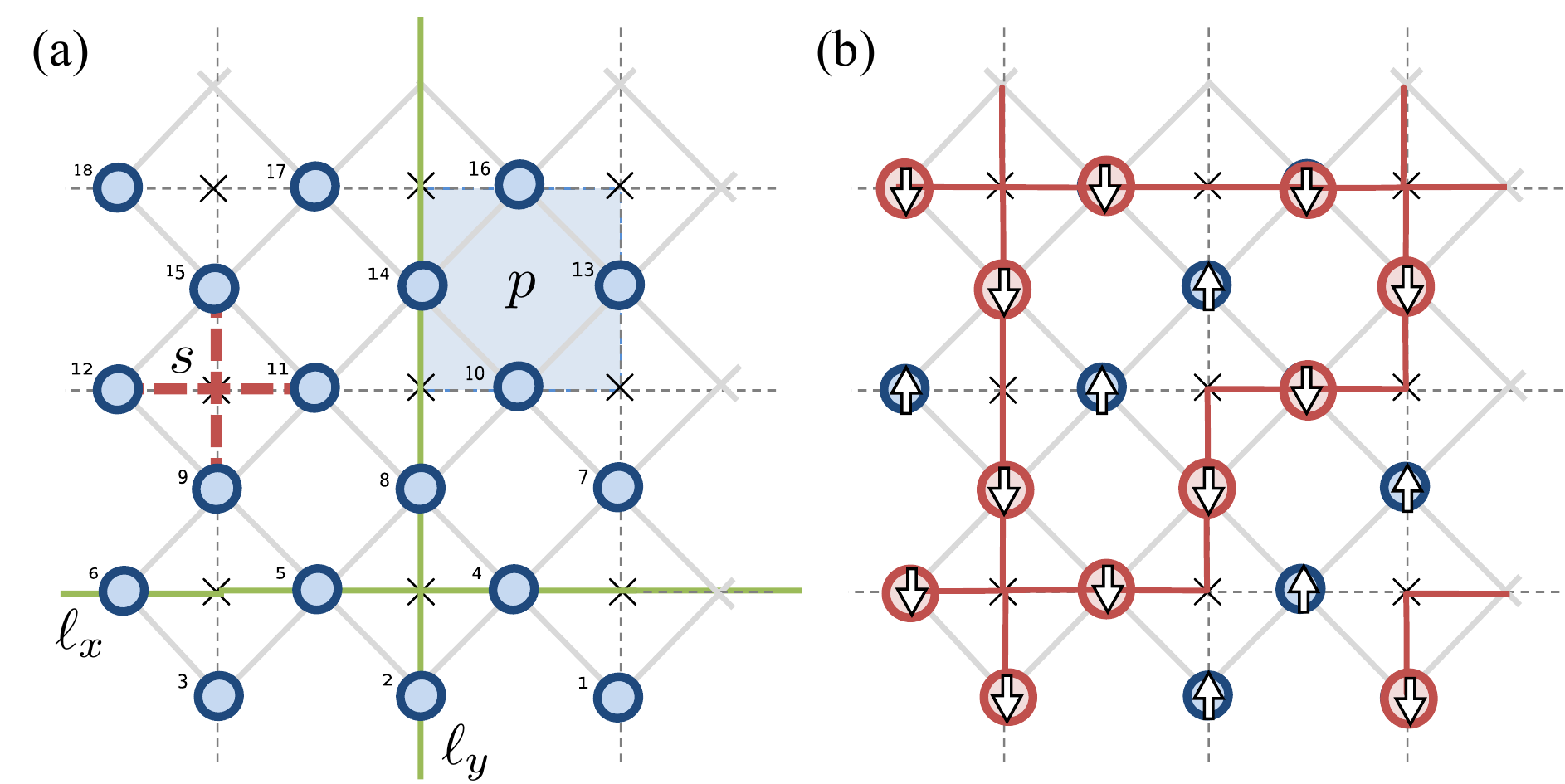}
\end{center}
\caption{(a) A $3 \times 3$ square lattice with vertices $\times$ forms a $3 \times 6$ lattice of spins (blue circles) located at the centre of bonds. Spins $s$ involved in one star $+$ surrounding a vertex are shown by the red dashed lines, while the spins $p$ involved in one plaquette $\square$ are shown as a blue shaded square. Possible non-contractible paths $\ell_{x,y}$ for the Wilson loop operators $\hat{\mathcal{W}}_{x,y}$ are also depicted as horizontal and vertical solid green lines. (b) A closed loop configuration state $\ket{\bm l}$.}
\label{fig:toric_ham}
\end{figure}

To be a +1 eigenstate of a given star term $\prod_{j \in s} \hat{Z}_j$ requires that an even number of the 4 spins in $s$ are $\ket{\downarrow}$. Given that neighbouring star terms overlap by a single spin a configuration state $\ket{\bm v}$ can be a simultaneous +1 eigenstate of all star terms $+$ so long as the pattern of $\ket{\downarrow}$ spins form closed loops around the lattice, which we denote as $\ket{\bm l}$. An example is shown in \fir{fig:toric_ham}(b). Given a closed loop state $\ket{\bm l}$ plaquette terms flip pairs of spins in the surrounding stars and so annihilate or create loops. We can form a +1 eigenstate of all plaquette terms $\square$ by equally superposing all closed loop configurations reachable from the chosen $\ket{\bm l}$. One of the simplest is the equal superposition of all possible closed loops\footnote{The degenerate subspace is spanned by the 4 states formed by equal superpositions of closed loops with the same $\pm 1$ winding number about the $x$- and $y$-axis. The state we consider is the equal superposition of these 4 states. Note that an arbitrary ground state in the degenerate subspace will not in general be a stabilizer state, and hence cannot be converted into an NQS using our method.}
\begin{equation}
\ket{\Psi_{\rm toric}} = \sum_{{\bm l}\in{\rm closed}}\ket{\bm l}. \label{eq:toric_state}
\end{equation}
This stabilizer state is uniquely specified by adding to the generator set two\footnote{When the $(2L^2+2) \times (2L^2)$ check matrix is put into standard form two zero rows appear that can be removed, leaving exactly $2L^2$ independent generators.} {\em Wilson} operators 
\begin{equation}
\hat{\mathcal{W}}_x = \prod_{j \in \ell_x} \hat{X}_j, \quad {\rm and} \quad \hat{\mathcal{W}}_y = \prod_{j \in \ell_y} \hat{X}_j,
\end{equation}
which involve spins on any non-contractible paths $\ell_x$ and $\ell_y$ that cut the lattice through vertices along the $x$- and $y$-axis, respectively, examples of which are shown in \fir{fig:toric_ham}(a). 

An NQS representation of $\ket{\Psi_{\rm toric}}$ can be constructed directly~\cite{clark_cps18}. Using the qubit labelling its non-zero amplitudes occur for configurations $\ket{\bm q}$ where $\bm q$ simultaneously satisfies $N/2$ binary equations ${\bf A}\,{\bm q}^{\rm T} = 0$ with the $N/2 \times N$ coefficient matrix $\bf A$ encoding $q_{s_1} \oplus q_{s_2} \oplus q_{s_3} \oplus q_{s_4} = 0$ for all sets of star spins $s_1,s_2,s_3,s_4$. Such constraints are easily enforced by a tensor network
\begin{equation}
\includegraphics[scale=0.5,valign=c]{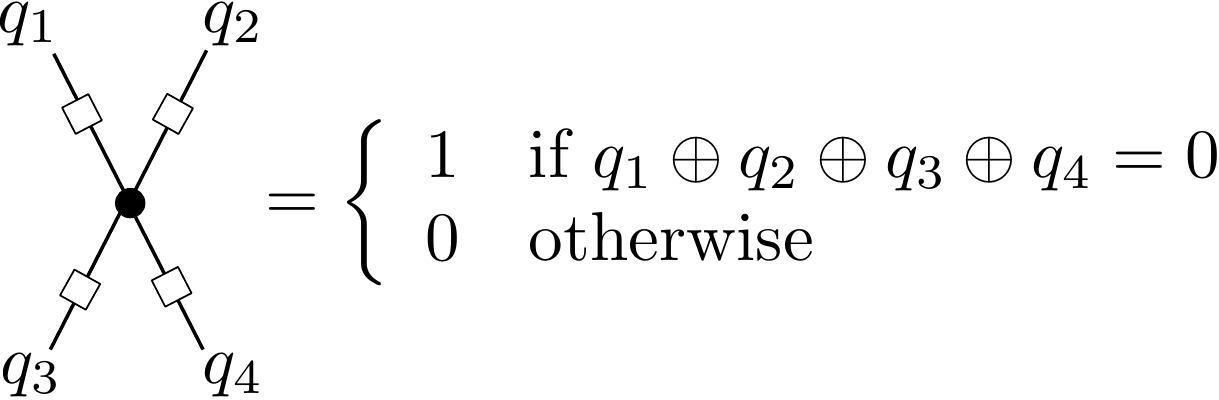} \quad , \label{eq:xor_tensor}
\end{equation}
commonly called the XOR tensor since its non-zero elements reflect the truth table of the 3 input bits and 1 output bit for a cascade of two XOR gates. The XOR tensor generalises straightforwardly for $n$ legs~\cite{denny11}. Given \eqr{eq:xor_tensor} is already in the form of a hidden unit with Hadamard coupling matrices, the NQS for $\ket{\Psi_{\rm toric}}$ is found by gluing an XOR tensor to the COPY tensors of spins in each star. This yields a direct NQS with $M=N/2$ hidden units, as shown in \fir{fig:toric_nqs}(a) for a $6 \times 3$ lattice. This representation elegantly reflects the translational invariance of the state, but as a consequence has no univalent visible unit COPY tensors, meaning it cannot emerge as a solution from the VMJ-NQS construction. 
\begin{figure}[htb]
\begin{center}
\includegraphics[scale=0.5]{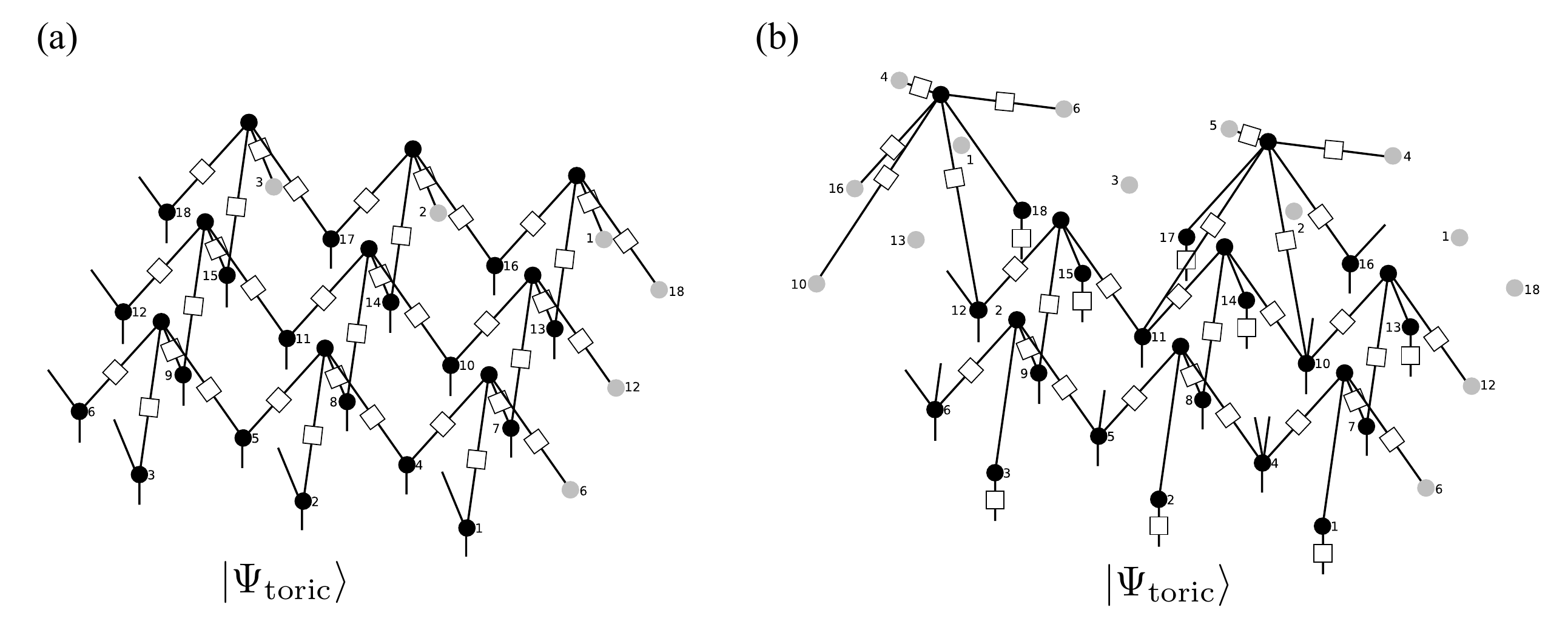}
\end{center}
\caption{(a) The direct NQS for $\ket{\Psi_{\rm toric}}$~\cite{clark_cps18} in which hidden units align with the star set $+$ for the lattice. Some visible units have open legs which indicate connections wrapping around the periodic boundaries displayed here as connections to duplicate spins shown as grey circles. (b) The VMJ-NQS representation of $\ket{\Psi_{\rm toric}}$ found through the construction procedure outlined in the proof of Theorem:~\ref{theorem:stab2nqs}. The resulting graph state has $\hat{\rm H}$ gates applied to the independent set $\mathcal{I}_{\rm H}(G) = \{1,2,3,13,14,15,17,18\}$ which also alone form a vertex cover for {\tt graph2nqs}.}
\label{fig:toric_nqs}
\end{figure}

The VMJ-NQS construction applied to $\ket{\Psi_{\rm toric}}$ has some general features. It identifies a graph state equivalence in which $\hat{\rm H}$ gates are applied to an independent set $\mathcal{I}_{\rm H}(G)$ that also forms a vertex cover. Consequently, the VMJ-NQS generated comprise entirely of Hadamard coupling matrices making each hidden unit an XOR tensor applied to receptive fields typically larger than 4 spins. The construction therefore reorganises the XOR constraints defining $\ket{\Psi_{\rm toric}}$ such that each each constraint has one spin that is exclusive to it that is a member of $\mathcal{I}_{\rm H}(G)$. This is equivalent to performing Gaussian elimination on the coefficient $\bf A$, and additionally reveals that only $N/2 - 1$ constraints are actually needed, since one can always be removed from the direct NQS by simply adding all the others to it. A particular VMJ-NQS solution with $M=8$ hidden units is shown in \fir{fig:toric_nqs}(b) for a $6 \times 3$ lattice. In this case it shares 6 hidden units with the direct NQS, but its final 2 hidden units have a coordination of 6. Given the considerable non-uniqueness of VMJ-NQS it is an interesting open question whether their structure can ever reflect the translational invariance of $\ket{\Psi_{\rm toric}}$, perhaps reduced over a larger unit cell.

\section{Conclusion and discussion} \label{sec:conclusion}
We have show that Jastrow states have an exact NQS representation with $M \sim O(N)$ hidden unit complexity that is substantially more efficient than the $M \sim O(N^2)$ previously conjectured. While this construction formally required diverging weights $w_{ij}$ it was seen to be robust to softening and a numerical example for the XXZ chain demonstrated that this new form of Jastrow NQS can emerge from an unbiased VMC optimisation. Focusing on the graph structure of Jastrow states, and using the tensor network formulation, we refined this result to show that Jastrow states possess an exact NQS representation requiring no more than $M=N-1$ hidden units, namely the largest size of a vertex cover of a graph. Moreover, these NQS representations are guaranteed to possess at least one univalent visible unit. This presented an opportunity to generalise them to VMJ-NQS by applying arbitrary single-spin gates to those units. By exploiting graph states, which are a special case of spin Jastrow states, this simple modification enhanced the expressiveness of VMJ-NQS sufficiently to capture all stabilizer states with no increase in the hidden unit complexity. Ultimately this result is a direct implication of the fact that both Jastrow and stabilizer states can all be generated by preparation circuits involving two-spin diagonal gates acting only on the visible units (and hence no ancilla) followed by non-diagonal single-spin unitary gates applied at independent vertices of the circuit's graph. 

There is good reason to suspect that wider classes of quantum states can be encapsulated by VMJ-NQS and hence share their compactness. For stabilizer states the construction relied on graph states, a special subset of Jastrow states, and local Clifford gates, a special subset of single spin unitaries. Yet VMJ-NQS are defined for {\em any} Jastrow state over a graph $G$ and can absorb {\em any} single-spin gates, unitary or otherwise, at independent vertices $G$, so very little of this generality was exploited. 

A good place to start are classes of quantum states closely related to graph and stabilizer states that are known to have efficient RBM representations~\cite{lu19}. This includes so-called {\em hypergraph} states~\cite{rossi13} which generalise standard graph states. Their construction again starts with all spins initialised in $\ket{+}$ but are now coupled via hyperedges where subsets of $p$ vertices $\mathcal{Z} = \{j_1,j_2,\dots,j_p\}$ have controlled$^{p-1}$-phase gates 
\begin{equation}
\hat{\rm C}^z_{j_1j_2\cdots j_p} = \mathbbm{1}_{j_1}\otimes\cdots\otimes\mathbbm{1}_{j_{p-1}}\otimes\mathbbm{1}_{j_p} + \mathbbm{1}_{j_1}\otimes\cdots\otimes\mathbbm{1}_{j_{p-1}}\hat{Z}_{j_p},
\end{equation}
applied to them. Each hyperedge thus introduces a factor of $-1$ if all $p$ spins are in the $\ket{\downarrow}$ state, equivalent to a correlation $\exp[({\rm i }\pi/2^p)(1 - v_{j_1})(1-v_{j_2})\cdots(1 - v_{j_p})]$. Despite hyperedges being many-body they are individually no more costly to describe than a standard $p=2$ edge. Specifically, for any value of $p$ they are captured by a single hidden unit\footnote{This solution is significantly more efficient than the $M=2p$ hidden units proposed in Ref.~\cite{lu19}.} with coupling matrices $\{{\bf C}^{(1)},{\bf C}^{(2)},\cdots,{\bf C}^{(N)}\}$
\begin{equation*}
\fl \qquad\qquad {\bf C}^{(j)} = \left[
\begin{array}{cc}
1 & 1   \\
0 & (-2)^{1/p}
\end{array}
\right] \quad {\rm if} ~ j \in \mathcal{Z}, \quad {\rm otherwise} \quad {\bf C}^{(j)} =  \left[
\begin{array}{cc}
1 & 1   \\
1 & 1
\end{array}
\right] .
\end{equation*}
However, for $p>2$ this hidden unit cannot be rearranged to expose a perfect correlation with one visible unit. As such an NQS for a hypergraph state built from these correlators does not benefit from the merging of COPY tensors used for standard graph states that allowed all edges originating from a vertex to be collected at a single hidden unit. This suggests that hypergraph states require a hidden unit for each hyperedge, so that for a fixed $p$ we have a complexity $M \sim O(N^p)$. Another related rich class of states are XS-stabilizer states~\cite{ni15}, which generalise stabilizer states to allow for non-Abelian stabilizers by drawing operators from the group generated by the single spin operators $\{\sqrt{\rm i}\mathbbm{1}, \hat{X},\hat{\rm S}\}$. Similar to how stabilizer states generalise graph states, XS-stabilizer states generalise the simplest $p=3$ hypergraph states by introducing nodal structure through parity constraints, giving them a formal hidden unit complexity of $M \sim O(N^3)$. It thus remains an open question whether compact NQS are a special case for Jastrow, graph and stabilizer states, or if they can also be found for hypergraph and XS-stabilizer states.

Further to this, the gauge symmetry reduction of tensor-network NQS outlined here has direct implications for the application of NQS to systems with higher on-site dimension. In particular it points to the need to go beyond the conventional complex RBM formulation. In recent work~\cite{pei21} we have considered this for the case of spin-1 systems and work in progress is examining it for bosonic systems~\cite{pei21b} where Jastrow states are a widely used variational ansatz. In this context it is an interesting open question whether our algorithm for constructing NQS for stabilizer states can be straightforwardly generalised for qudits~\cite{hostens05}.

Finally, our work gives useful guidance about the structure of NQS beyond the exact states considered. For example, in numerical calculations using the complex RBM parameterisation we have found that allowing moderate-valued weights enables the optimisation to locate near perfect visible-hidden correlations that can improve the accuracy, even when the exact state is not a Jastrow state. Consequently, if the value of weights are heavily constrained during an optimisation then such compact solution might be missed. We have also seen new patterns of receptive fields emerge, like a decreasing coordination pattern interpolating between system-extensive and sparse connectivity, which could be exploited in numerical calculations. Moreover, our results have demonstrated that even $\alpha = 1$ NQS can exactly capture highly non-trivial quantum states that exhibit volume-scaling entanglement and topological order. These observations contribute to the perpetual balance within variational approaches between using a more specialised ansatz with fewer parameters, but more bias, verses systematically increasing the parameters in an ansatz to improve expressiveness, but risk complicating the energy landscape traversed by the optimisation. 

\section*{Acknowledgements}
SRC gratefully acknowledges support from the UK's Engineering and Physical Sciences Research Council (EPSRC) under grants EP/P025110/2 and EP/T028424/1. MP acknowledges the University of Bristol Advanced Computing Research Centre for the use of their High Performance Computing facility (BlueCrystal Phase 3 and 4) in performing the VMC calculations presented.
\newpage

\appendix

\section{Boltzmann-like parameterisation of coupling matrices} \label{app:boltzmann}
Every $2 \times 2$ coupling matrix ${\bf C}^{(ij)}$ can be parameterised locally in Boltzmann form as
\begin{equation}
C^{(ij)}_{h_iv_j} = \exp\left(c_{ij} + w_{ij}h_iv_j + \tilde{b}_{ij}h_i + \tilde{a_{ij}}v_j\right), \label{eq:rbm_coupling_mat}
\end{equation}
for $h_i, v_j \in \{+1,-1\}$ in terms of a weight $w_{ij}$, partial biases $\tilde{a}_{ij}$ and $\tilde{b}_{ij}$, and a scale factor $c_{ij}$. These complex parameters are found from the coupling matrix elements as
\begin{eqnarray*}
\tilde{a}_{ij} &=& \frac{1}{4}\left[\log\left(C^{(ij)}_{++}\right) - \log\left(C^{(ij)}_{+-}\right) + \log\left(C^{(ij)}_{-+}\right) - \log\left(C^{(ij)}_{--}\right)\right], \\
\tilde{b}_{ij} &=& \frac{1}{4}\left[\log\left(C^{(ij)}_{++}\right) + \log\left(C^{(ij)}_{+-}\right) - \log\left(C^{(ij)}_{-+}\right) - \log\left(C^{(ij)}_{--}\right)\right], \\
c_{ij} &=& \frac{1}{4}\left[\log\left(C^{(ij)}_{++}\right) + \log\left(C^{(ij)}_{+-}\right) + \log\left(C^{(ij)}_{-+}\right) + \log\left(C^{(ij)}_{--}\right)\right], \\
w_{ij} &=& \frac{1}{4}\left[\log\left(C^{(ij)}_{++}\right) - \log\left(C^{(ij)}_{+-}\right) - \log\left(C^{(ij)}_{-+}\right) + \log\left(C^{(ij)}_{--}\right)\right].
\end{eqnarray*}
As discussed in \secr{sec:xxz} the zeros appearing in coupling matrix elements can be handled numerically by softening $C^{(ij)}_{h_iv_j} \mapsto \max(C^{(ij)}_{h_iv_j},e^{-\mathcal{S}})$, where $\mathcal{S} \approx 5-10$, to avoid divergent parameters.

Using this decomposition we can immediately remove an overall irrelevant constant $e^{c_{ij}}$ from each coupling matrix, reducing the total number of complex parameters of the NQS tensor network to $3MN$. The decomposition
\begin{eqnarray*}
\fl\qquad \left[
\begin{array}{cc}
C^{(ij)}_{++} & C^{(ij)}_{+-}   \\
C^{(ij)}_{-+} & C^{(ij)}_{--}
\end{array}
\right] &\simeq& \left[
\begin{array}{cc}
e^{\tilde{b}_{ij}}  & 0  \\
0  & e^{-\tilde{b}_{ij}}   
\end{array}
\right] \left[
\begin{array}{cc}
e^{w_{ij}} & e^{-w_{ij}}  \\
e^{-w_{ij}} & e^{w_{ij}}  
\end{array}
\right]\left[
\begin{array}{cc}
e^{\tilde{a}_{ij}}  & 0  \\
0  & e^{-\tilde{a}_{ij}}   
\end{array}
\right],
\end{eqnarray*}
is represented diagrammatically as
\begin{equation}
\includegraphics[scale=0.5,valign=c]{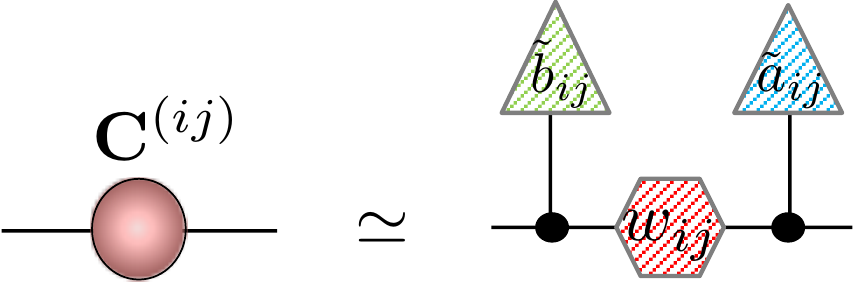} , \label{eq:cmat_decomp}
\end{equation}
in terms of exponentially parameterised diagonal matrices sandwiching a Boltzmann matrix (emphasised by hashed shading). The fusion rule \eqr{eq:copy_fusion} results in the elementwise multiplication of diagonal matrices so their exponents get summed at common COPY tensor
\begin{equation}
\includegraphics[scale=0.5,valign=c]{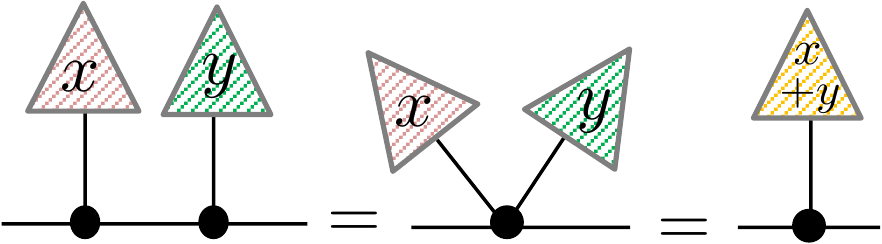} . \label{eq:diag_biases}
\end{equation}
Hence, by decomposing all $MN$ coupling matrices the $2MN$ partial biases are summed to give $M$ hidden biases $b_i = \sum_{j=1}^N \tilde{b}_{ij}$ and $N$ visible biases $a_j = \sum_{i=1}^M \tilde{a}_{ij}$~\cite{chen18}. We thus arrive at an NQS tensor network built from Boltzmann-type components
\begin{equation}
\fl\qquad\includegraphics[scale=0.5,valign=c]{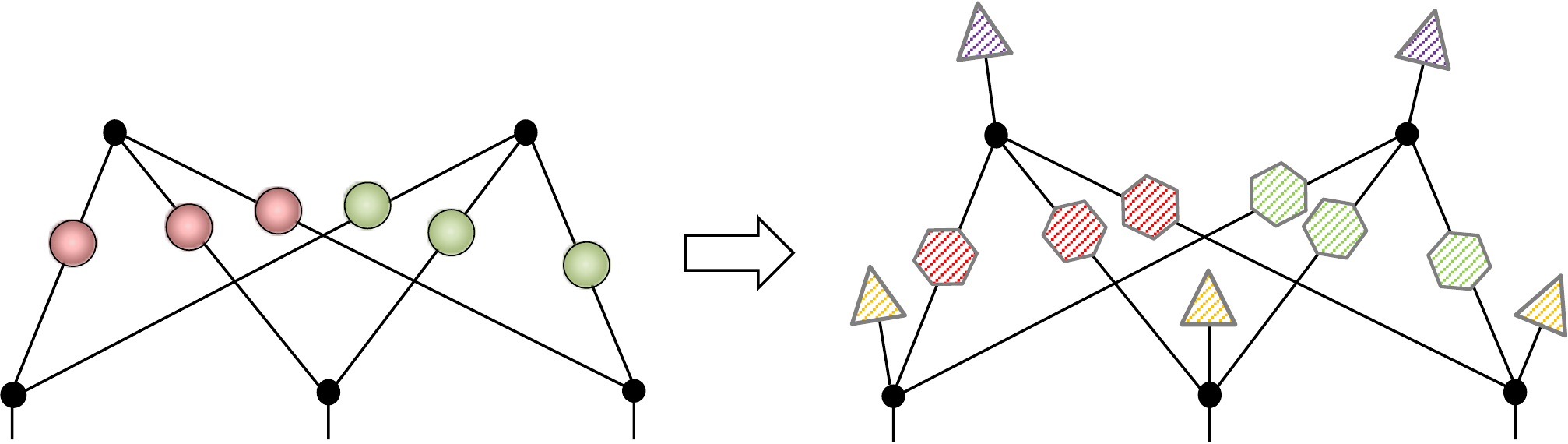} ,
\end{equation}
defined by $MN + M + N$ complex parameters, entirely equivalent to the complex RBM formulation in \eqr{rbm}. 

\section*{References}
\bibliographystyle{iopart-num}
\bibliography{nqs_graph}

\end{document}